\documentclass[aps,pra,reprint,superscriptaddress]{revtex4-2}
\usepackage{amsmath}
\usepackage{amssymb}
\usepackage{amsthm}
\usepackage{braket}
\usepackage{physics}
\usepackage{mathtools}
\usepackage[caption=false]{subfig}
\usepackage{bm}
\usepackage{tikz}
\usetikzlibrary{shapes.geometric, arrows}
\tikzstyle{process} = [rectangle, 
minimum width=3cm, 
minimum height=1cm, 
text centered, 
text width=8cm, 
draw=black, 
fill=orange!30]
\tikzstyle{process_on} = [rectangle, 
minimum width=3cm, 
minimum height=1cm, 
text centered, 
text width=8cm, 
draw=black, 
fill=green!30]
\tikzstyle{arrow} = [thick,->,>=stealth]
\theoremstyle{definition}
\newtheorem{cj}{Conjecture}
\newtheorem{theorem}{Theorem}
\newtheorem{col}{Corollary}
\newcommand{\arank}{\mathrm{arank}}
\newcommand{\brank}{\mathrm{brank}}
\newcommand{\crank}{\mathrm{crank}}
\newcommand{\wrank}{\mathrm{wrank}}
\usepackage[maxfloats=256]{morefloats}
\begin{document}

\title{Implementing arbitrary multi-mode continuous-variable quantum gates\\with fixed non-Gaussian states and adaptive linear optics}

\author{Fumiya Hanamura}
\affiliation{Department of Applied Physics, School of Engineering, The University of Tokyo, 7-3-1 Hongo, Bunkyo-ku, Tokyo 113-8656, Japan}
\author{Warit Asavanant}
\affiliation{Department of Applied Physics, School of Engineering, The University of Tokyo, 7-3-1 Hongo, Bunkyo-ku, Tokyo 113-8656, Japan}
\affiliation{Optical Quantum Computing Research Team, RIKEN Center for Quantum Computing, 2-1 Hirosawa, Wako, Saitama 351-0198, Japan}
\author{Hironari Nagayoshi}
\affiliation{Department of Applied Physics, School of Engineering, The University of Tokyo, 7-3-1 Hongo, Bunkyo-ku, Tokyo 113-8656, Japan}
\author{Atsushi Sakaguchi}
\affiliation{Optical Quantum Computing Research Team, RIKEN Center for Quantum Computing, 2-1 Hirosawa, Wako, Saitama 351-0198, Japan}
\author{Ryuhoh Ide}
\affiliation{Department of Applied Physics, School of Engineering, The University of Tokyo, 7-3-1 Hongo, Bunkyo-ku, Tokyo 113-8656, Japan}
\author{Kosuke Fukui}
\affiliation{Department of Applied Physics, School of Engineering, The University of Tokyo, 7-3-1 Hongo, Bunkyo-ku, Tokyo 113-8656, Japan}
\author{Peter van Loock}
\affiliation{Institute of Physics, Johannes-Gutenberg University of Mainz, Staudingerweg 7, 55128 Mainz, Germany}
\author{Akira Furusawa}
\affiliation{Department of Applied Physics, School of Engineering, The University of Tokyo, 7-3-1 Hongo, Bunkyo-ku, Tokyo 113-8656, Japan}
\affiliation{Optical Quantum Computing Research Team, RIKEN Center for Quantum Computing, 2-1 Hirosawa, Wako, Saitama 351-0198, Japan}

\date{\today}

\begin{abstract}
Non-Gaussian quantum gates are essential components for optical quantum information processing. However, the efficient implementation of practically important multi-mode higher-order non-Gaussian gates has not been comprehensively studied. We propose a measurement-based method to directly implement general, multi-mode, and higher-order non-Gaussian gates using only fixed non-Gaussian ancillary states and adaptive linear optics. Compared to existing methods, our method allows for a more resource-efficient and experimentally feasible implementation of multi-mode gates that are important for various applications in optical quantum technology, such as the two-mode cubic quantum non-demolition gate or the three-mode continuous-variable Toffoli gate, and their higher-order extensions. Our results will expedite the progress toward fault-tolerant universal quantum computing with light.
\end{abstract}


\maketitle

\section{Introduction}
Continuous-variable (CV) optical systems are a promising platform for large-scale quantum information processing. Alongside the large-scale Gaussian operations enabled by cluster states \cite{cluster,mikkel_cluster}, non-Gaussian operations are crucial elements \cite{universal,qec_nogo} for many practical tasks ranging from fault-tolerant universal quantum computing \cite{gkp} to quantum simulation \cite{quantum_simulation}. The single-mode cubic phase gate (CPG) \cite{gkp}, whose Hamiltonian is $\hat{x}^3$, has been intensively explored \cite{gkp,miyata_cpg} as the simplest example of such non-Gaussian operations, and has been demonstrated in a proof-of-principle experiment \cite{sakaguchi_cpg}.
However, for practical tasks such as generation and manipulation of code words for quantum error correction \cite{gkp}, or quantum simulation of complex quantum systems \cite{quantum_simulation}, multi-mode and/or higher-order non-Gaussian gates are required. 

Direct deterministic implementations of such non-Gaussian gates are challenging in optical systems because of their small intrinsic nonlinearity. Thus, ancilla-assisted implementations with offline probabilistic ancilla states are a common approach \cite{miyata_cpg,sakaguchi_cpg}. However, in most proposals, only the implementation of elementary single-mode gates such as CPGs is discussed, and multi-mode, higher-order gates are implemented via decomposition of the gates into multiple CPGs and Gaussian operations \cite{cpg_gkp,decomp_gate,exact_decomposition}. This requires a large number of CPGs and the noise from each gate accumulates, making the composed gate noisy. As a complementary approach, a method has been proposed to directly perform the higher-order and/or multi-mode gates, using higher-order and/or multi-mode non-Gaussian ancillary states and nonlinear feedforward \cite{multimode_coupling}. This scheme has the advantage that it requires fewer steps to implement the gate, at the expense of more complex non-Gaussian ancillary states. However, this proposal is somewhat incomplete, as it requires an adaptive preparation of different non-Gaussian states depending on previous measurement outcomes, whose experimental implementation is not trivial.

In this work, we propose a general methodology to implement high-order multi-mode non-Gaussian gates, requiring only the offline preparation of fixed ancillary states that depend solely on the gate to be implemented, and adaptive linear optics, which are experimentally concrete and feasible resources. Our protocol is measurement-based \cite{mbqc}, thus compatible with quantum information processing using cluster states \cite{cluster}. Our protocol allows different choices of non-Gaussian ancillary states for implementing the same gate. Exploiting this degree of freedom, we propose some heuristic approaches based on Chow decomposition \cite{chow_rank} of polynomials to reduce the number of ancillary modes. We apply our method to several important examples, including the cubic quantum non-demolition (QND) gate \cite{cpg_gkp}, which is a lowest-order non-Gaussian entangling gate, and the CV Toffoli gate \cite{cpg_gkp}, which provides a minimal universal gate set together with the Hadamard gate, and their higher-order extensions. We show that both in 3rd-order cases and higher-order cases, one can reduce the number of required Gaussian and non-Gaussian ancillary modes compared to the conventional decomposition into single-mode non-Gaussian gates \cite{cpg_gkp,exact_decomposition,seth_decomp}. Our results open up new possibilities for CV quantum circuit optimization, which not only expedites the progress toward fault-tolerant universal quantum computing and efficient quantum simulation, but also leads to a better understanding of complex multi-mode quantum dynamics.

The structure of the paper is as follows. In Sec.~\ref{sec:notation}, we introduce some preliminary notations used throughout this paper. In Sec.~\ref{sec:decomp_quad}, we discuss that an arbitrary gate can be decomposed into ``quadrature gates'', whose Hamiltonian only includes one of the orthogonal quadrature operators. In Sec.~\ref{sec:mbqc}, we introduce the measurement-based implementation of quadrature gates, showing the equivalence of measurements and gates. In Sec.~\ref{sec:3rd_order}, we propose implementations of multi-mode 3rd-order gates, introducing the concept of generalized linear coupling, which provides the freedom to choose different non-Gaussian ancillary states. In Sec.~\ref{sec:3rd_order_examples}, we describe some examples of 3rd-order gates and their implementations. In Sec.~\ref{sec:higher_order}, we describe the general methodology of how to implement higher-order gates. In Sec.~\ref{sec:reduction}, we propose some heuristic approaches to reduce the number of ancillary modes using mathematical tools such as Chow decomposition. In Sec~\ref{sec:general_examples}, we present some examples of higher-order gates. We demonstrate the resource-efficiency of our scheme compared to conventional schemes by using the strategies to reduce the ancillary modes introduced in Sec.~\ref{sec:reduction}.

\section{Definitions and Notations}\label{sec:notation}
The variables $x$ and $p$ represent quadrature operators of an optical mode, which satisfy a commutation relation $[x,p]=i$. Hats ($\hat{\cdot}$) of operators are omitted whenever it is clear from the context. We define the quadrature operator of arbitrary phase $\theta$ as
\begin{align}
    p_{\theta}=p\cos\theta+x\sin\theta.
\end{align}
Bold symbols such as $\bm{s}$ and $\bm{x}$ represent vectors of either c-numbers or operators.

For diagramatic notations, we use the notation illustrated in Fig.~\ref{fig:multimode_notation} for multi-mode states in circuit diagrams. We define a mode-wise beamsplitter for two parts of multi-mode states, with arbitrary numbers of modes $n,n'$ for each part. It is characterized by a vector of amplitude transmittances $\bm{t}=(t_1,\dots,t_{\min(n,n')})$. We define it as in Fig.~\ref{fig:n_np}, thus when $n\neq n'$, some of the modes just pass through without interaction.

We define a multi-mode beamsplitter described by an orthogonal matrix $O$ as an operation $U$ such that
\begin{align}
    U^\dagger \bm{x}U&=O\bm{x},\\
    U^\dagger \bm{p}U&=O\bm{p}.
\end{align}
We depict this as a rectangle with $O$ as in Fig.~\ref{fig:multimode_bs}. 
\begin{figure}[ht]
    \centering
    \subfloat[]{\includegraphics[width=0.7\linewidth]{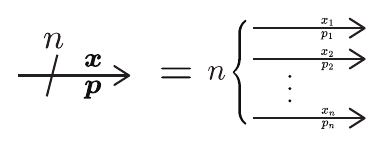}\label{fig:multimode_notation}}\\
    \subfloat[]{\includegraphics[width=0.7\linewidth]{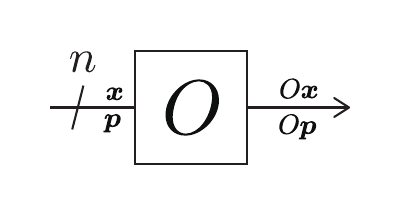}\label{fig:multimode_bs}}\\
    \subfloat[]{\includegraphics[width=0.8\linewidth]{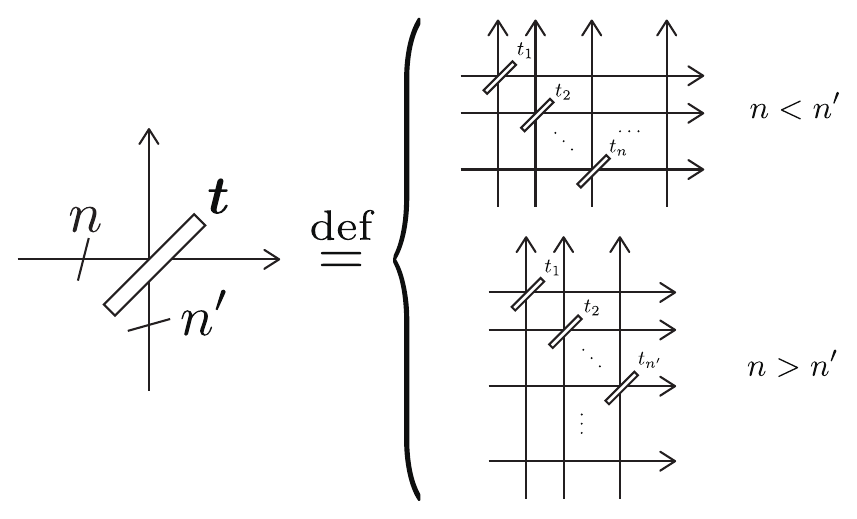}\label{fig:n_np}}\\
    \caption[]{\subref{fig:multimode_notation} Notation for $n$-mode states in circuit diagrams. \subref{fig:multimode_bs} Multi-mode beamsplitter corresponding to an orthogonal matrix $O$. \subref{fig:n_np} Mode-wise beamsplitter with transmittance $\bm{t}$.}
    \label{fig:notation}
\end{figure}

For multivariable polynomials, we use the tensor notation
\begin{align}
f(\bm{x})&=\sum_{k=0}^{N}f^{(k)}\bm{x}^{\otimes k}\\
&=\sum_{k=0}^{N}\sum_{i_1,\dots,i_k} f^{(k)}_{i_1,\dots,i_k}x_{i_1}\dots x_{i_k}
\end{align}
where $f^{(k)}$ is a symmetric tensor of rank $k$, corresponding to the $k$-th order coefficient of $f(\bm{x})$. For a homogeneous polynomial $f$, we sometimes just denote $f^{(N)}$ as $f$, by a slight abuse of notation. 
\section{Gate decomposition into quadrature gates}\label{sec:decomp_quad}
We consider a general multi-mode unitary operation
\begin{align}
    \exp[iH(x_1,\dots,x_n)]
\end{align}
with an arbitrary Hamiltonian $H$ in the form of a finite-order polynomial:
\begin{align}
    H(x_1,\dots,x_n)=\sum_i \qty(c_i\prod_k x_k^{m_{ik}}p_k^{n_{ik}}+c_i^*\prod_k p_k^{n_{ik}}x_k^{m_{ik}})
\end{align}
Using Trotter-Suzuki approximation \cite{trotter_suzuki,decomp_gate}, this can be decomposed into linear operations and Hamiltonians of the form
\begin{align}
    \prod_k x_k^{m_{k}}.
\end{align}
This only includes $x$-quadratures. We shall refer to this as a quadrature gate. Thus, in this form, it is directly suitable for a measurement-based implementation. In the rest of this paper, we will focus on how to realize these multi-mode quadrature gates. 

Note that such decomposition is not unique for a given Hamiltonian. In App.~\ref{sec:decomp_quad_appendix}, we explicitly provide an example of such decomposition for an arbitrary Hamiltonian, as a generalization of the method in Ref.~\cite{decomp_gate}.

\section{Measurement-based implementation of quadrature gates}\label{sec:mbqc}
\begin{figure}[ht]
    \centering
    \includegraphics[width=0.9\linewidth]{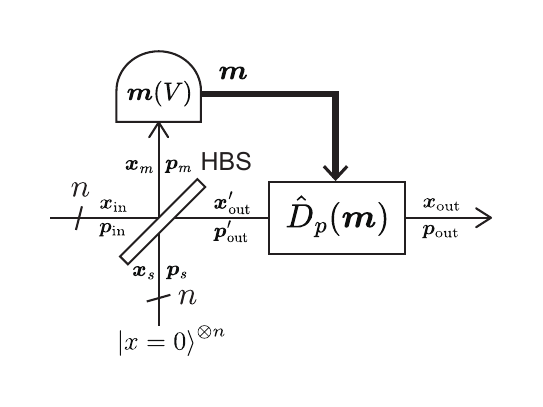}
    \caption{Measurement-based implementation of the Hamiltonian $V(x_1, \dots, x_n)$, with a constant squeezing factor. This setup includes ancillary squeezed states, a mode-wise half beamsplitter (HBS, as illustrated in Fig.~\ref{fig:multimode_bs}), measurements of nonlinear operators $\bm{m}(V)$, and feedforwarded displacement $\hat{D}_p(\bm{m})$ to the $p$ quadrature based on the measurement outcomes $\bm{m}$.}
    \label{fig:mbqc}
\end{figure}

We consider a quadrature gate $U=\exp[iV(x_1,\dots,x_n)]$ with an arbitrary $x$-Hamiltonian $V$ in the form of a finite-order polynomial:
\begin{align}
V(\bm{x})&=\sum_{k=1}^{N}V^{(k)}\bm{x}^{\otimes k}.
\end{align}

We consider a measurement-based implementation of the gate $U$. We define $n$ commuting operators
\begin{align}
\begin{split}
m_i(\bm{x},\bm{p};V)&=U^\dagger p_i U\\
&=p_i-\frac{\partial}{\partial x_i}V(x_1,\dots,x_n).
\end{split}
\end{align}
We write this as
\begin{align}
    \bm{m}(\bm{x},\bm{p};V)&=\bm{p}-\frac{\partial}{\partial \bm{x}}V(\bm{x})
\end{align}
in short. We also just write this as $\bm{m}(V)$, when $\bm{x},\bm{p}$ are clear from the context.

Then we consider the circuit in Fig.~\ref{fig:mbqc}, which consists of $n$ modes of ancillary squeezed states ($x$-eigenstates), a half beamsplitter, and simultaneous measurements of $m_i$. The quadrature operators after the beamsplitter are written as
\begin{align}
    \bm{x}_{\mathrm{out}}'&=\frac{1}{\sqrt{2}}(\bm{x}_{\mathrm{in}}+\bm{x}_{s}),\\
    \bm{p}_{\mathrm{out}}'&=\frac{1}{\sqrt{2}}(\bm{p}_{\mathrm{in}}+\bm{p}_{s}),\\
    \bm{x}_{m}&=\frac{1}{\sqrt{2}}(\bm{x}_{\mathrm{in}}-\bm{x}_{s}),\\
    \bm{p}_{m}&=\frac{1}{\sqrt{2}}(\bm{p}_{\mathrm{in}}-\bm{p}_{s}).
\end{align}
By measuring the operator $\bm{m}(\bm{x}_m,\bm{p}_m;V)$ and performing a displacement to $\bm{p}'_{\mathrm{out}}$ by a feedforward, the output quadratures are expressed as
\begin{align}
    \bm{x}_{\mathrm{out}}&=\frac{1}{\sqrt{2}}(\bm{x}_{\mathrm{in}}+\bm{x}_{s}),\label{eq:mbqc_x_with_noise}\\
    \bm{p}_{\mathrm{out}}&=\bm{p}_{\mathrm{out}}'+\bm{m}(\bm{x}_m,\bm{p}_m;V)\\
    &=\sqrt{2}\bm{p}_{\mathrm{in}}-\sqrt{2}\frac{\partial}{\partial \bm{x}_{\mathrm{in}}}V\qty(\frac{1}{\sqrt{2}}(\bm{x}_{\mathrm{in}}-\bm{x}_s))\label{eq:mbqc_p_with_noise}.
\end{align}
Because we take the $x$ eigenstate as the ancillary state, $\bm{x}_s$ can be substituted by $0$, obtaining
\begin{align}
    \bm{x}_{\mathrm{out}}&=\frac{1}{\sqrt{2}}\bm{x}_{\mathrm{in}},\label{eq:mbqc_x_transform}\\
    \bm{p}_{\mathrm{out}}&=\sqrt{2}\bm{p}_{\mathrm{in}}-\sqrt{2}\frac{\partial}{\partial \bm{x}_{\mathrm{in}}}V\qty(\frac{1}{\sqrt{2}}\bm{x}_{\mathrm{in}}).\label{eq:mbqc_p_transform}
\end{align}
This can be written as
\begin{align}
    \bm{x}_{\mathrm{out}}&=U^\dagger S^\dagger \bm{x}_{\mathrm{in}}SU,\\
    \bm{p}_{\mathrm{out}}&=U^\dagger S^\dagger \bm{p}_{\mathrm{in}}SU,
\end{align}
with a squeezing operator $S$ satisfying
\begin{align}
    S^\dagger \bm{p}S&=\sqrt{2}\bm{p},\\
    S^\dagger \bm{x}S&=\frac{1}{\sqrt{2}}\bm{x}.
\end{align}
This means the circuit implements the operation $U$ up to constant squeezing $S$.

Thus, the problem of implementing the gate is now reduced to finding an implementation for the simultaneous measurement of the non-Gaussian operators $m_i$. From here on, we focus exclusively on this measurement-based model. Therefore, when we refer to the implementation of a gate $V$, we mean the implementation of a measurement of $\bm{m}(V)$. In the following sections, we demonstrate that this measurement can be performed using multi-mode non-Gaussian ancillary states and linear optics, and we discuss methods to minimize the required resources for this implementation.

Note that, in a realistic situation where we use finitely squeezed states instead of the $x$-eigenstates as the ancillary states, the output state has additional classical Gaussian displacement noises expressed by $\bm{x}_s$, as can be seen from Eqs.~(\ref{eq:mbqc_x_with_noise}) and (\ref{eq:mbqc_p_with_noise}).

\section{Implementation of 3rd-order Hamiltonians}\label{sec:3rd_order}
In order to see how a non-Gaussian measurement can be implemented using non-Gaussian ancillary states and nonlinear feedforward, we first consider the simplest case of 3rd-order Hamiltonian. Suppose we want to implement a Hamiltonian
\begin{align}
    V(x_1,\dots,x_n)&=V\bm{x}^{\otimes 3}\\
    &=\sum_{ijk} V_{ijk} x_i x_j x_k,
\end{align}
where $V_{ijk}$ is an arbitrary symmetric tensor of rank 3. In order to implement this operation using the measurement-based method explained in Sec. \ref{sec:mbqc}, one needs to measure the operators
\begin{align}
    \bm{m}&=\bm{m}(\bm{x}_m,\bm{p}_m;V)\\
    &=\bm{p}_m-3V\bm{x}_m^{\otimes 2}.\label{eq:3rd_order_measurement}
\end{align}

Below, we introduce two methods for implementing this measurement. The first method involves mode-wise coupling using mode-specific beamsplitters and adaptive homodyne measurements, which require minimal resources for implementation. The second method generalizes this scheme using the concept of generalized linear coupling that we introduce. In this approach, a virtual Gaussian operation can be adaptively performed on the ancillary state using only adaptive linear optics. This second method enhances the flexibility of the implementation and is crucial for the implementation of higher-order gates discussed in Sec.~\ref{sec:higher_order}.
\begin{figure}[ht]
    \centering
    \subfloat[]{\includegraphics[width=0.9\linewidth]{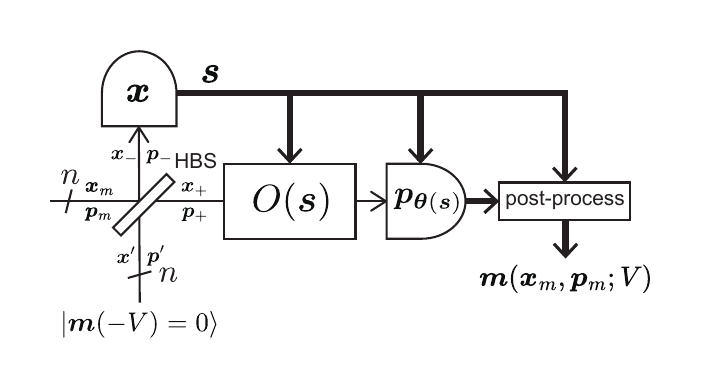}\label{fig:3rd_order_hbs}}\\
    \subfloat[]{\includegraphics[width=0.9\linewidth]{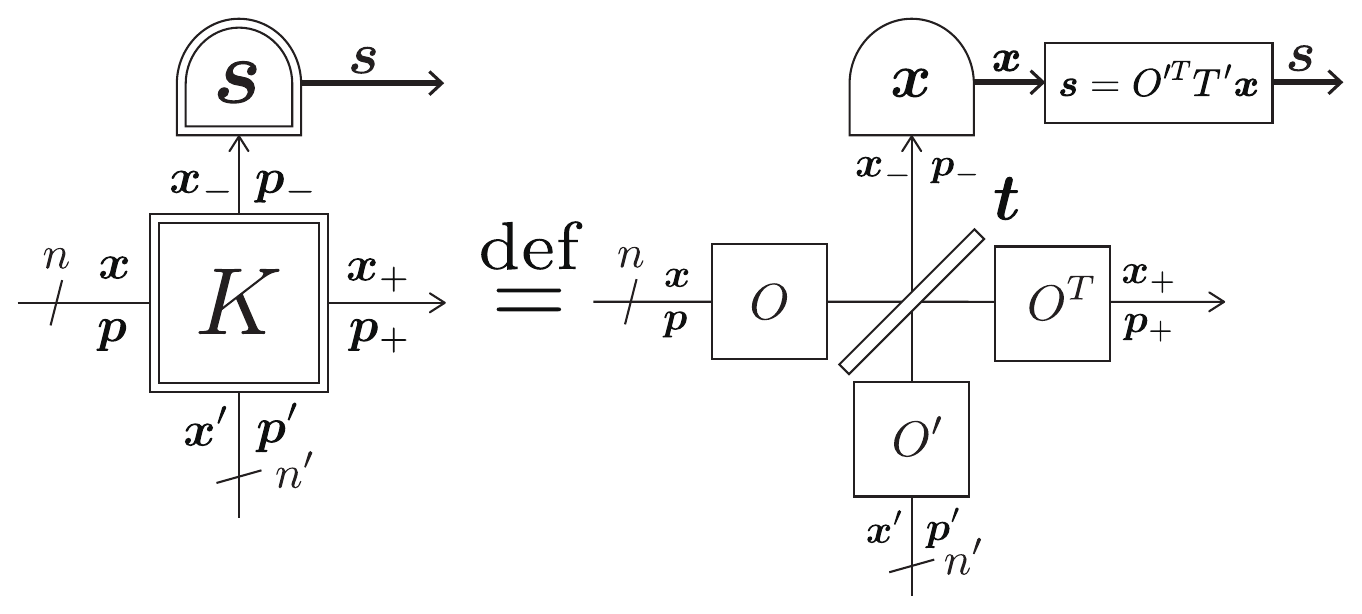}\label{fig:building_block}}\\
    \subfloat[]{\includegraphics[width=0.9\linewidth]{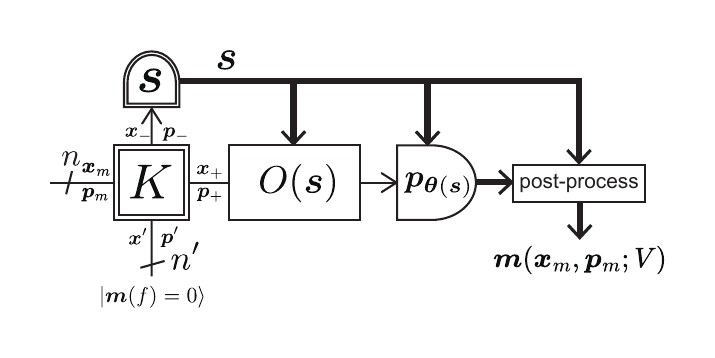}\label{fig:3rd_order_general}}
    \caption[]{\subref{fig:3rd_order_hbs} Implementation of a 3rd-order Hamiltonian with mode-wise coupling. The multi-mode beamsplitter $O(\bm{s})$ and the phase $\bm{\theta}$ are determined by the homodyne measurement outcomes $\bm{s}$. Note that the input field $\bm{x}_m,\bm{p}_m$ corresponds to the measured field in Fig.~\ref{fig:mbqc}, and the scheme implements the measurement of $\bm{m}(V)$. \subref{fig:building_block} Definition of generalized linear coupling characterized by a matrix $K$ and measurement $\bm{s}$. We draw them as a double square and a double-line detector, respectively. \subref{fig:3rd_order_general} Implementation of a 3rd-order Hamiltonian with generalized linear coupling $K$. The function $f$ corresponding to the ancillary states should satisfy Eq.~(\ref{eq:ancilla_cond}).}
\end{figure}
\subsection{Mode-wise coupling}\label{sec:mode_wise}
We first consider the circuit in Fig.~\ref{fig:3rd_order_hbs}. The input state is combined with an $n$-mode non-Gaussian ancillary state using mode-wise half-beamsplitters. The quadrature operators of the input and the ancillary modes are denoted as $\bm{x}_m,\bm{p}_m$ and $\bm{x}',\bm{p}'$, respectively. The quadrature operators after the beamsplitters are expressed as
\begin{align}
    \bm{x}_{-}&=\frac{1}{\sqrt{2}}(\bm{x}_m-\bm{x}'),
    &\bm{p}_{-}&=\frac{1}{\sqrt{2}}(\bm{p}_m-\bm{p}'),\\
    \bm{x}_{+}&=\frac{1}{\sqrt{2}}(\bm{x}_m+\bm{x}'),
    &\bm{p}_{+}&=\frac{1}{\sqrt{2}}(\bm{p}_m+\bm{p}').
\end{align}
When we define
\begin{align}
    \bm{m}'&=\bm{m}(\bm{x}',\bm{p}';-V)\\
    &=\bm{p}'+3V\bm{x}'^{\otimes 2},
\end{align}
and
\begin{align}
    \bm{\delta}&=\bm{m}+\bm{m}',
\end{align}
we obtain
\begin{align}
    \bm{\delta}&=\bm{p}_m+\bm{p}'-3V(\bm{x}_m^{\otimes 2}-\bm{x}'^{\otimes 2})\\
    &=\sqrt{2}\bm{p}_{+}-6V (\bm{x}_{+}\otimes\bm{x}_{-}).
\end{align}

Then we measure $\bm{x}_{-}$ and suppose we get the outcomes $\bm{x}_{-}=\bm{s}$. This makes the measurement of $\bm{\delta}$ equivalent to the measurement of the linear quadrature operators
\begin{align}
    \bm{\delta}&=\sqrt{2}\bm{p}_{+}-6V (\bm{s}\otimes\bm{x}_{+})\\
    &=\sqrt{2}(\bm{p}_{+}+A\bm{x}_{+})\label{eq:linear_quad}
\end{align}
where we define a rank-2 tensor (matrix)
\begin{align}
    A=-3\sqrt{2}V\bm{s}.
\end{align}

Thus $\bm{\delta}$ is a set of commuting linear combinations of the quadrature operators. These operators can be simultaneously measured using a multi-mode beamsplitter followed by homodyne measurements on each mode, as shown in App.~\ref{sec:linear_measurement}. Hence the measurement of $\bm{\delta}$ can be performed using the circuit in Fig.~\ref{fig:3rd_order_hbs}. The orthogonal matrix $O(\bm{s})$ and the phases of the homodyne measurement $\theta_i(\bm{s})$ are determined via the diagonalization of $A$ as follows,
\begin{align}
A&=O^T\mqty(\tan{\theta_1}&&\\&\ddots&\\&&\tan{\theta_n})O.\label{eq:evd}
\end{align}
Since these parameters are nonlinear functions of $\bm{s}$, nonlinear feedforward is required for implementing the measurement.

When we choose the ancillary state to be the eigenstate of $\bm{m}'$ satisfying
\begin{align}
    \bm{m}'=0,
\end{align}
the measurement of $\bm{\delta}$ is equivalent to the measurement of $\bm{m}$. In the presence of imperfections of the ancillary state, $\bm{m}'$ introduces extra noise to the measurement.
\subsection{Generalized linear coupling}\label{sec:generalized_coupling}
In Ref.~\cite{general_npg}, it is mentioned that one can effectively control the ancilla squeezing by changing the transmittance of the beamsplitter. This is useful because it can save squeezing resources using only linear optics, and it allows for adaptive squeezing operations on the ancillary state, which are required for the implementation of higher-order gates, as we will discuss in Sec.~\ref{sec:higher_order}. We generalize this idea for multi-mode cases, resulting in a structure that enables a broader range of multi-mode Gaussian operations on the ancillary state, not just mode-wise squeezing.

Suppose we have $n$ input modes and $n'$ ancillary modes. For any $n\times n'$ matrix $K$, let
\begin{align}
    K&=O^{\prime T}\Sigma O
\end{align}
be the singular value decomposition of $K$. $O$ and $O'$ are $n\times n$ and $n'\times n'$ orthogonal matrices, respectively, and $\Sigma$ is a $n'\times n$ diagonal matrix. Figure \ref{fig:building_block} defines a generalized linear coupling characterized by $K$ as a combination of three multi-mode beamsplitters $O,O',O^T$ and a mode-wise beamsplitter $\bm{t}$, where $\bm{t}$ is determined so that
\begin{align}
    \frac{r_i}{t_i}&=\lambda_i,&i=1,\dots,\min(n,n')
\end{align}
for the diagonal elements $\lambda_i$ of $\Sigma$ (see Sec.~\ref{sec:notation} and Fig.~\ref{fig:notation} for the diagrammatic notations used in Fig.~\ref{fig:building_block}). We express it as a double square as in Fig.~\ref{fig:building_block}. The quadrature operators of input and ancillary modes are denoted as $\bm{x},\bm{p}$ and $\bm{x}',\bm{p}'$. The quadrature operators of the output modes are denoted as $\bm{x}_{-},\bm{p}_{-}$ and $\bm{x}_{+},\bm{p}_{+}$. We also define a measurement $\bm{s}$ expressed by a double-line detector in Fig.~\ref{fig:building_block} as a measurement of $\bm{s}=O'^TT'\bm{x}_{-}$, where
\begin{align}
    T'&=\begin{cases}
    \mathrm{diag}(t_1,\dots,t_{n'}) & \mathrm {if\ } n'\leq n\\
    \mathrm{diag}(t_1,\dots,t_{n},1,\dots,1) & \mathrm {otherwise}
\end{cases}
\end{align}
This can be done either by first measuring $\bm{x}_{-}$ and postprocessing the measurement outcomes, or by measuring $\bm{x}_{-}$ after applying the multi-mode beamsplitter $O'^TT'$. Then, we have the following theorem.
\begin{theorem}\label{thm:diamond_relation}
For arbitrary $f,g$, there is a relation
\begin{align}
\begin{split}
&P(K)\bm{m}(\bm{x},\bm{p};f)+KP(K)\bm{m}(\bm{x}',\bm{p}';g)\\&=\bm{m}(\bm{x}_{+},\bm{p}_{+};f\diamond_{K,\bm{s}}g),
\end{split}
\end{align}
where we define a matrix $P(K)$ as
\begin{align}
    P(K)=(I+K^TK)^{-1/2}
\end{align}
and a binary operator $\diamond_{K,\bm{s}}$ between two functions $f,g$ as
\begin{align}
    (f \diamond_{K,\bm{s}} g)(\bm{x})=f(P(K)\bm{x} + K^T \bm{s})+g(K P(K)\bm{x} - \bm{s}).
\end{align}
\end{theorem}
The proof of Theorem \ref{thm:diamond_relation} is given in App.~\ref{sec:generalized_coupling_appx}. For a 3rd-order Hamiltonian $V$ and a 3rd-order function $f$, we have
\begin{align}
    (V \diamond_{K,\bm{s}} f)(x)&=V(P(K)\bm{x} + K^T \bm{s})^{\otimes 3}+f(K P(K)\bm{x} - \bm{s})^{\otimes 3}\\
    &=(V+fK^{\otimes 3})(P(K)\bm{x})^{\otimes 3}+O(\bm{x}^2)
\end{align}
Thus, if one finds a function $f$ such that
\begin{align}
    V+fK^{\otimes 3}=0,\label{eq:ancilla_cond_3}
\end{align}
$V \diamond_{K,\bm{s}} f$ becomes a 2nd-order function of $\bm{x}$, and so the measuremennt of $\bm{m}(\bm{x}_{\mathrm{out}},\bm{p}_{\mathrm{out}};V \diamond_{K,\bm{s}} f)$ can be performed using beamsplitters and homodyne measurement using the same method as in Sec.~\ref{sec:mode_wise}. Therefore, if we use the eigenstate of $\bm{m}(\bm{x}',\bm{p}';f)$ as the ancillary states, the measurement of $\bm{m}(\bm{x}_{\mathrm{in}},\bm{p}_{\mathrm{in}};V)$ can be implemented with the scheme in Fig.~\ref{fig:3rd_order_general}.

From the form of the condition Eq.~(\ref{eq:ancilla_cond_3}), this generalized linear coupling effectively applies a multi-mode Gaussian operation
\begin{align}
    \bm{x}\to K\bm{x}
\end{align}
to the input state. Note that the case of mode-wise HBS coupling described in Sec.~\ref{sec:mode_wise} corresponds to the case where $K=I$.
\begin{figure}[ht]
    \centering
    \subfloat[]{\includegraphics[width=1.0\linewidth]{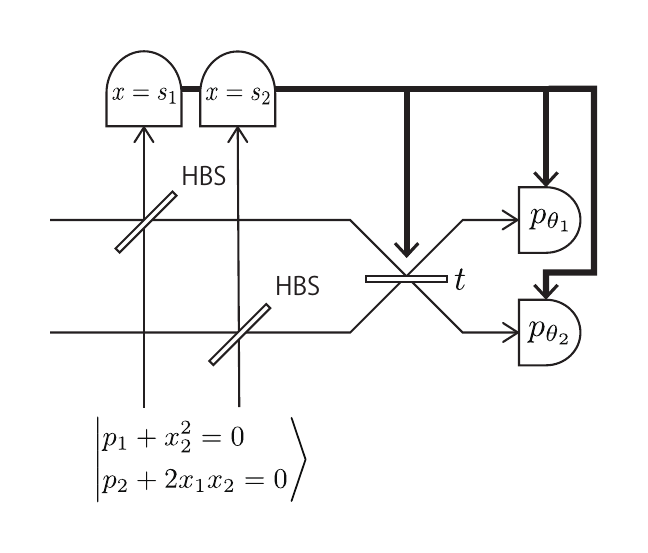}\label{fig:cubic_qnd}}\\
    \subfloat[]{\includegraphics[width=1.0\linewidth]{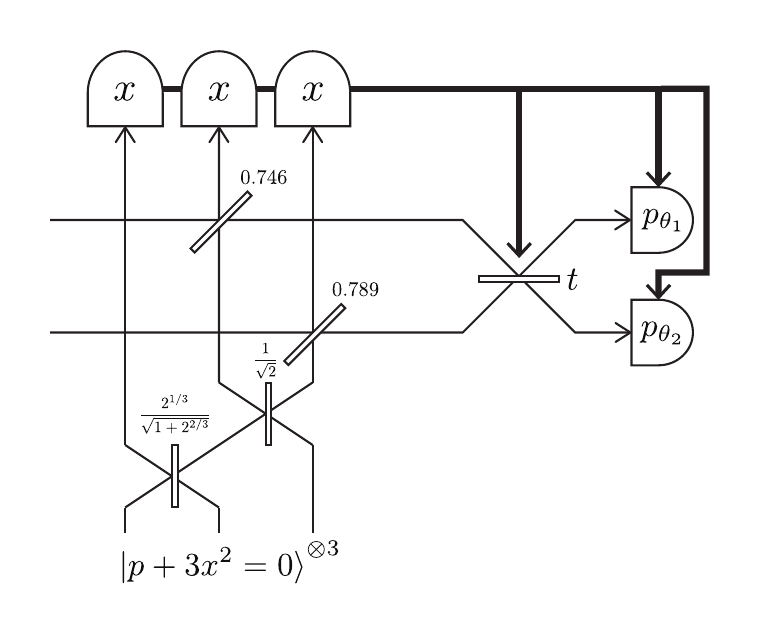}\label{fig:cubic_qnd_general}}
    \caption[]{Implementation of the cubic-QND gate. \subref{fig:cubic_qnd} When mode-wise coupling is used. \subref{fig:cubic_qnd_general} When single-mode ancillary states and generalized linear coupling are used.}
\end{figure}
\begin{figure}[ht]
    \centering
    \subfloat[]{\includegraphics[width=1.0\linewidth]{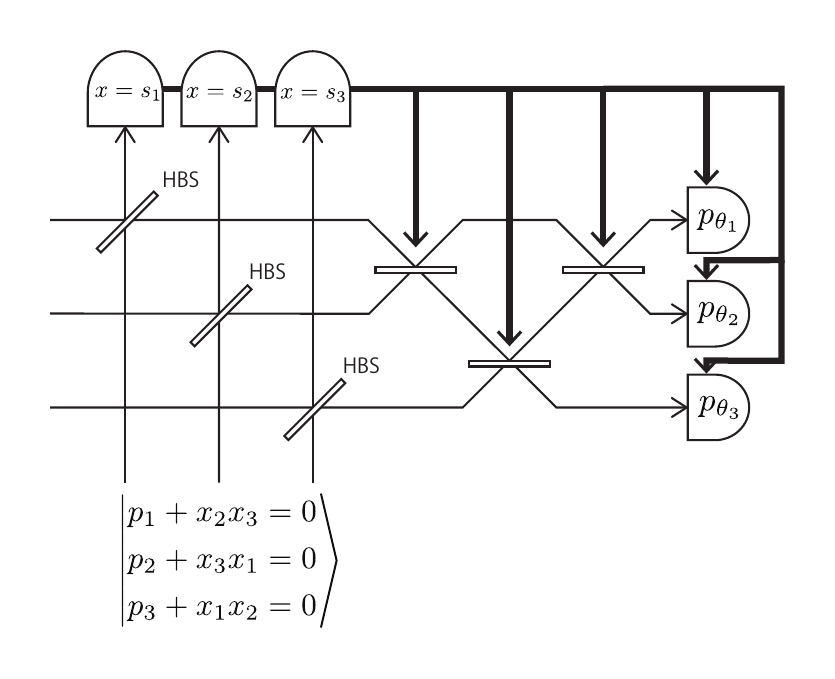}\label{fig:toffoli}}\\
    \subfloat[]{\includegraphics[width=1.0\linewidth]{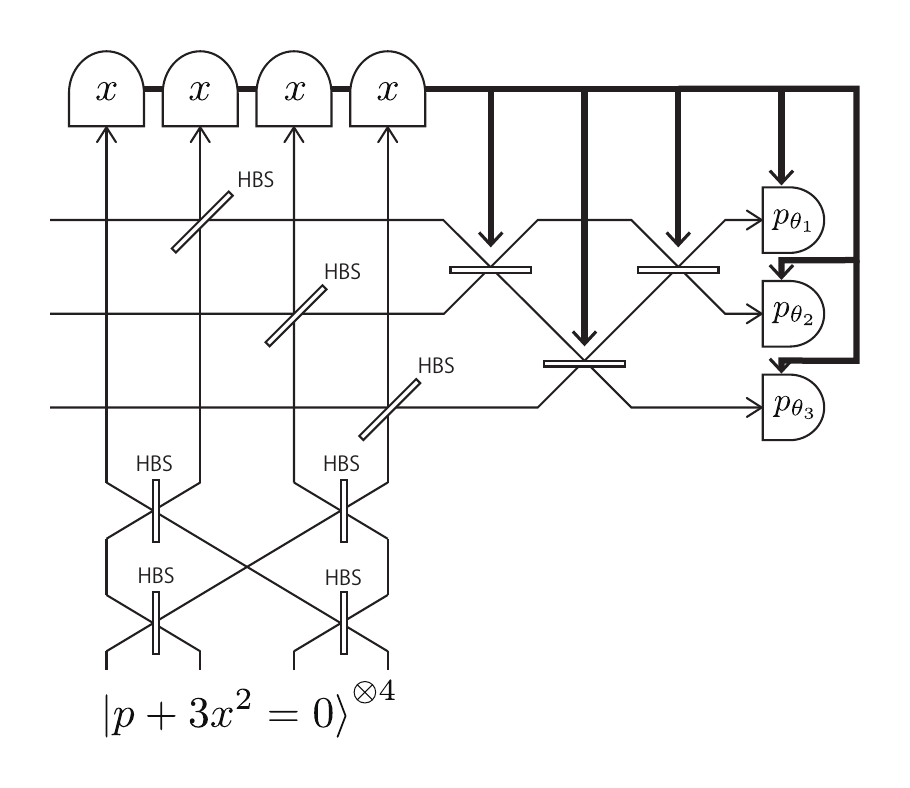}\label{fig:toffoli_general}}
    \caption[]{Implementations of Toffoli gates. \subref{fig:toffoli} When mode-wise coupling is used. \subref{fig:toffoli_general} When single-mode ancillary states and generalized linear coupling are used.}
\end{figure}
\section{Examples of 3rd-order gates}\label{sec:3rd_order_examples}
\subsection{Cubic-QND gate}\label{sec:cubic_qnd}
The simplest example of a 3rd-order multi-mode Hamiltonian is the cubic-QND gate, whose Hamiltonian can be written as
\begin{align}
    V(x_1,x_2)=x_1x_2^2.\label{eq:cubic-qnd}
\end{align}
This can be implemented using the scheme in Fig.~\ref{fig:cubic_qnd} which consists of half beamsplitters, feedforwarded variable beamsplitters (VBSs), homodyne measurements, and a two-mode ancillary state satisfying
\begin{align}
\begin{split}
    &p_1+x_2^2=0,\\
    &p_2+2x_1x_2=0. 
\end{split}\label{eq:cubic-qnd-ancilla}
\end{align}
Note that a physical approximation of this state is discussed in Ref.~\cite{multimode_coupling}. The transmittance of the variable beamsplitter $t$ and the phases $\theta_i$ of the homodyne measurement are determined from the measurement outcomes $s_1,s_2$, by performing the following eigenvalue decomposition:
\begin{align}
    \sqrt{2}\mqty(0&s_2\\s_2&s_1)&=\mqty(t&-r\\r&t)\mqty(\tan\theta_1&0\\0&\tan\theta_2)\mqty(t&r\\-r&t)
\end{align}

An alternative way of implementing the gate can be obtained by rewriting Eq.~(\ref{eq:cubic-qnd}) as
\begin{align}
    V(x_1,x_2)=\frac{1}{6}\qty[(x_1+x_2)^3+(x_1-x_2)^3-2x_1^3]
\end{align}
(such a decomposition is called \textit{Waring decomposition} \cite{waring}, see also Sec.~\ref{sec:waring}). From this, we can take
\begin{align}
    f(x_1,x_2,x_3)=-(x_1^3+x_2^3+x_3^3)
\end{align}
and
\begin{align}
    K=\frac{1}{\sqrt[3]{6}}\mqty(1&1\\1&-1\\-\sqrt[3]{2}&0)
\end{align}
in the scheme of Fig.~\ref{fig:3rd_order_general}.
This $K$ has a SVD
\begin{align}
    K=O'\cdot \frac{1}{\sqrt[3]{6}}\mqty(\sqrt{2+2^{-2/3}}&0\\0&\sqrt{2}\\0&0),
\end{align}
where the orthogonal matrix $O'$ can be expressed as a product of two two-mode beamsplitter matrices:
\begin{align}
\begin{split}
O'=&\mqty(\frac{1}{\sqrt{2}}&-\frac{1}{\sqrt{2}}&0\\\frac{1}{\sqrt{2}}&\frac{1}{\sqrt{2}}&0\\0&0&1) \mqty(\frac{2^{1/3}}{\sqrt{1 + 2^{2/3}}}& 0& \frac{1}{\sqrt{1 + 2^{2/3}}}\\ 0& 1& 0\\ -\frac{1}{\sqrt{1 + 2^{2/3}}}&0& \frac{2^{1/3}}{\sqrt{1 + 2^{2/3}}}).
\end{split}
\end{align}
Thus, the gate can be implemented using the circuit in Fig.~\ref{fig:cubic_qnd_general}. Here three cubic-phase states (CPSs) $\ket{p+3x^2=0}^{\otimes 3}$ are used as the ancillary state.

In the first case using the two-mode ancilla of Eq.~(\ref{eq:cubic-qnd-ancilla}), the number of ancillary modes (two non-Gaussian, two Gaussian) is less than for the decomposition of the gate into three CPGs, which requires three non-Gaussian and three Gaussian ancillary modes \cite{cpg_gkp}. Even in the second case using three CPSs as the ancillary states, the number of the Gaussian ancillary modes is reduced to two without changing the non-Gaussian ancillary states.

\subsection{Toffoli gate}
Another important example is the CV Toffoli gate \cite{cpg_gkp}
\begin{align}
    V(x_1,x_2,x_3)=x_1x_2x_3
\end{align}
which provides a universal gate set together with the Hadamard gate. The Toffoli gate can be implemented with the scheme in Fig.~\ref{fig:toffoli}, using a three-mode ancillary state satisfying
\begin{align}
\begin{split}
    &p_1+x_2x_3=0,\\
    &p_2+x_3x_1=0,\\
    &p_3+x_1x_2=0,
\end{split}
\end{align}
and a three-mode VBS, which can be realized using three two-mode VBSs. The orthogonal matrix $O$ corresponding to the VBS and the phases $\theta_i$ of the homodyne measurements can be obtained from the following eigenvalue decomposition:
\begin{align}
\sqrt{2}\mqty(0&s_3&s_2\\s_3&0&s_1\\s_2&s_1&0)=O^T\mqty(\tan\theta_1&0&0\\0&\tan\theta_2&0\\0&0&\tan\theta_2)O
\end{align}

This implementation of the Toffoli gate requires three non-Gaussian and three Gaussian ancillary modes, which is a reduced number of modes compared to a known decomposition of the gate into four CPGs, requiring four non-Gaussian and four Gaussian ancillary modes \cite{cpg_gkp}. Similary to the example of the cubic-QND gate in Sec.~\ref{sec:cubic_qnd}, one can also use four CPSs as the ancillary state, keeping the number of squeezed ancillary states at three, as in Fig.~\ref{fig:toffoli_general}. This is because the Waring decomposition of $V$ is given by
\begin{align}
\begin{split}
    V(x)=&\frac{1}{24}[(x_1+x_2+x_3)^3+(-x_1-x_2+x_3)^3\\
    &+(-x_1+x_2-x_3)^3+(x_1-x_2-x_3)^3],
\end{split}
\end{align}
thus one can take
\begin{align}
    &f(\bm{x})=x_1^3+x_2^3+x_3^3+x_4^3,\\
    K&=\frac{1}{\sqrt[3]{24}}\mqty(1&1&1\\-1&-1&1\\-1&1&-1\\1&-1&-1)=O'\cdot\frac{1}{\sqrt[3]{3}}\mqty(1&0&0\\0&1&0\\0&0&1\\0&0&0),\\
O'&=\mqty(1&1&0&0\\-1&1&0&0\\0&0&1&1\\0&0&-1&1)\mqty(1&0&0&0\\0&0&1&0\\0&1&0&0\\0&0&0&1)\mqty(1&1&0&0\\-1&1&0&0\\0&0&1&1\\0&0&-1&1).
\end{align}
\section{Implementation of higher-order Hamiltonians}\label{sec:higher_order}
\begin{figure}[ht]
\centering
\begin{tikzpicture}[node distance=2cm]
\node (star_seq) [process] {Find $f_1,\dots,f_m, K_1,\dots,K_m$ satisfying Eq.~(\ref{eq:star_chain}).\\(Sec.~\ref{sec:reduction})};
\node (dia_seq) [process, below of=star_seq] {Obtain $f_1,\dots,f_m, K_1,\dots,K_m$ and $A(s),b(s),c(s)$ satisfying Eq.~(\ref{eq:diamond_chain}). Build the system in Fig.~\ref{fig:general_case}.\\
(Theorem \ref{thm:diamond_star_seq})};
\node (kff) [process_on, below of=dia_seq] {Measure $\bm{s}_{1},\dots, \bm{s}_{i}$ and perform non-linear feedforward to the VBSs $K_{i+1}$\\
(Sec.~\ref{sec:reduction})};
\node (ff) [process_on, below of=kff] {Using $\bm{s}_{1},\dots, \bm{s}_{m}$, perform non-linear feedforward to VBS $O(s)$ and homodyne phases $\theta_i$.\\
(Eq.~(\ref{eq:evd}), App.~\ref{sec:linear_measurement})};
\node (post_process) [process_on, below of=ff] {Perform homodyne measurement $\bm{p}_{\bm{\theta}}$, post-process $\bm{p}_{\bm{\theta}}$ to obtain the measurement outcomes of $\bm{m}(V)$.\\
(Eq.~(\ref{eq:evd}), App.~\ref{sec:linear_measurement})
};

\draw [arrow] (star_seq) -- (dia_seq);
\draw [arrow] (dia_seq) -- (kff);
\draw [arrow] (kff) -- (ff);
\draw [arrow] (ff) -- (post_process);
\draw node [left,below of=post_process] {\colorbox{orange!30}{offline}\ \ \ \colorbox{green!30}{online}};
\end{tikzpicture}
\caption{Procedure to implement the measurement $\bm{m}(V)$ corresponding to the general multi-mode Hamiltonian $V$ (Eq.~(\ref{eq:general_hamiltonian})). For the details, see Secs.~\ref{sec:higher_order} and \ref{sec:reduction}.}
\label{fig:procedure}
\end{figure}
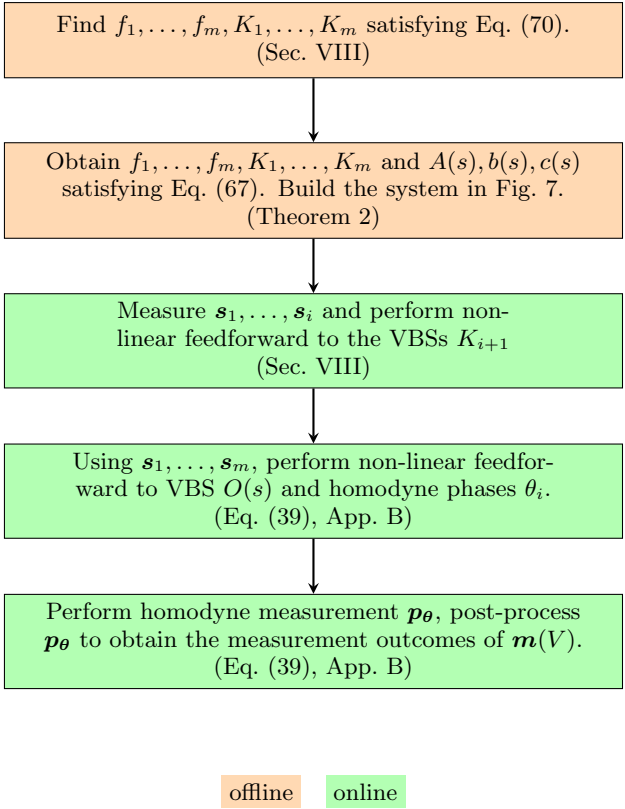

In this section, we consider the general $N$-th order Hamiltonian
\begin{align}
V(\bm{x})&=\sum_{k=1}^{N} V^{(k)}\bm{x}^{\otimes k},\label{eq:general_hamiltonian}
\end{align}
for $n$ modes ($\bm{x}=(x_1,\dots,x_n)$). Ref.~\cite{general_npg} shows that the measurement corresponding to a single-mode quadrature phase gate $e^{ix^N}$ of an arbitrary order $N$ can be implemented using $N-2$ non-Gaussian ancillary states and linear optics, by sequentially decreasing the order of the measured polynomial using nonlinear feedforward operations. We generalize this idea to the multi-mode case Eq.~(\ref{eq:general_hamiltonian}). In Fig.~\ref{fig:procedure}, we summarize the whole procedure to implement the measurement of $\bm{m}(V)$, which is explained in this and the following section.
\begin{figure*}[ht]
    \centering
   \includegraphics[width=1.0\linewidth]{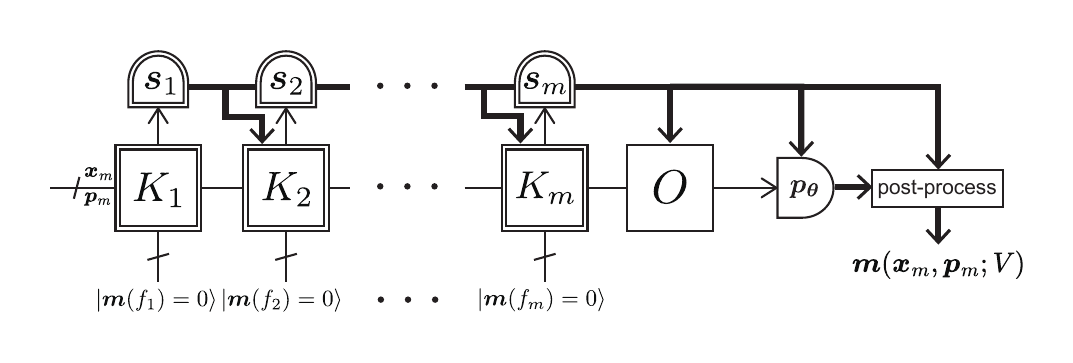}
    \caption{Implementation of higher-order Hamiltonian, consisting of multiple non-Gaussian ancillary states, linear optics (generalized linear coupling, Sec.\ref{sec:generalized_coupling}), and homodyne measurements. Part of the linear optics and measurements are adaptively changed depending on previous measurement outcomes.}
    \label{fig:general_case}
\end{figure*}

From Theorem \ref{thm:diamond_relation}, if one can find $f_1,\dots,f_m$, $K_1,\dots,K_m$ and $A(s),b(s),c(s)$ satisfying
\begin{align}
\begin{split}
    &(V \diamond_{K_1,\bm{s}_1} f_1 \diamond_{K_2,\bm{s}_2} \dots \diamond_{K_m,\bm{s}_m} f_m)(x)\\
    &=x^TA(s)x+b(s)^Tx+c(s)\label{eq:diamond_chain}
\end{split}
\end{align}
for arbitrary measurement outcomes $\bm{s}_i$, then the measurement of $\bm{m}(\bm{x},\bm{p};V)$ can be implemented using the scheme in Fig.~\ref{fig:general_case}. Here the beamsplitter matrix $O$ and the phases of the homodynes $\theta_i$ are determined via the eigenvalue decomposition of $A(s)$ (Eq.~(\ref{eq:evd})).
For finding such a sequence, it is convenient to define an operator $\star_{K,\bm{s}}$ as
\begin{align}
    (f\star_{K,\bm{s}} g)(\bm{x})=f(\bm{x})+g(K\bm{x}-\bm{s}),
\end{align}
because we have the following theorem.

\begin{theorem}\label{thm:diamond_star_seq}
For any sequences $f_1,\dots,f_m$, $K_1,\dots,K_m$ there exist $K'_1,\dots,K'_m$ and $A$ such that for any $s_1,\dots,s_m$, there exist $s'_1,\dots,s'_m$ and $b$ satisfying
\begin{align}
\begin{split}
    &(V \diamond_{K'_1,s_1} f_1 \diamond_{K'_2,s_2} \dots \diamond_{K'_m,s_m} f_m)(\bm{x})\\
    =&(V \star_{K_1,s_1'} f_1 \star_{K_2,s_2'} \dots \star_{K_m,s_m'} f_m)(A\bm{x}+b)
\end{split}\label{eq:chain_equivalence}
\end{align}
\end{theorem}

The proof of Theorem \ref{thm:diamond_star_seq} is given in App.~\ref{sec:diamond_star_seq}. From Theorem \ref{thm:diamond_star_seq}, it is sufficient to find a sequence $f_k,K_k$ and $A,b,c$ such that
\begin{align}
    (V \star_{K_1,s_1} f_1 \star_{K_2,s_2} \dots \star_{K_m,s_m} f_m)(x)=x^TAx+b^Tx+c.\label{eq:star_chain}
\end{align}
More specifically, if we write 
\begin{align}
    (V \star_{K_1,s_1} f_1 \star_{K_2,s_2} \dots \star_{K_{j},s_{j}} f_{j})(x)&=\sum_{k=1}^{N_j}V^{(j,k)}x^{\otimes k},\label{eq:v_coeffs}
\end{align}
where $N_j=\mathrm{deg}(V \star_{K_1,s_1} f_1 \star_{K_2,s_2} \dots \star_{K_{j},s_{j}} f_{j})$, we have
\begin{align}
V^{(j+1,k)}=V^{(j,k)}+\sum_{l=k}^{N_j}\mqty(l\\k)f_{j+1}^{(l)}(-s)^{\otimes l-k}K_{j+1}^{\otimes k} \label{eq:poly_reduction_general}.
\end{align}
In particular, we have
\begin{align}
V^{(j+1,N_j)}=V^{(j,N_j)}+f_{j+1}^{(N_j)}K_{j+1}^{\otimes N_j}.
\end{align}
Thus if one takes $f_{j+1}$ and $K_{j+1}$ so that 
\begin{align}
V^{(j,N_j)}+f_{j+1}^{(N_j)}K_{j+1}^{\otimes N_j}=0,\label{eq:cond_order_reduction}
\end{align}
one has $N_j=N-j$ and Eq.~(\ref{eq:star_chain}) can be satisfied with $m=N-2$ steps. When $f_{j+1}$ is taken to be a homogeneous polynomial of order $N-j$ satisfying
\begin{align}
    V^{(j,N_j)}+f_{j+1}K_{j+1}^{\otimes N_j}=0,\label{eq:ancilla_cond}
\end{align}
Eq.~(\ref{eq:poly_reduction_general}) becomes
\begin{align}
V^{(j+1,k)}=V^{(j,k)}+\mqty(N_j\\k)f_{j+1}(-s)^{\otimes N_j-k}K_{j+1}^{\otimes k} \label{eq:poly_reduction}.
\end{align}

If one can adaptively prepare the ancillary non-Gaussian states depending on the measurement outcomes $s_i$, the whole process can be implemented using only $(N-2)\times n$ ancillary modes by choosing
\begin{align}
    f_k=-V^{(k-1,N+1-k)},K_k=I.
\end{align}
However, in an actual setup, it is often difficult to prepare non-Gaussian states adaptively, and a better strategy is to prepare fixed ancillary states $f_k$ and adaptively change $K_k$. When this strategy is taken, $f_k$ should not depend on the measurement outcomes $s_1,\dots, s_{k-1}$. In the following sections, we consider this situation.

\section{Reduction of the number of ancillary modes}\label{sec:reduction}
In this section, we consider the problem to minimize the number of the non-Gaussian ancillary modes, when those states cannot be adaptively prepared depending on the previous measurement outcomes $s$. Although the general solution for finding the global minimum is still an open question, we give several observations and heuristic strategies. In Sec.~\ref{sec:general_examples}, we apply those strategies to some examples and compare the performances.

We recommend the reader to first check the example in Sec.~\ref{sec:small_example} before reading the following discussion, for getting an intuition about our idea.
\subsection{Sign problem}\label{sec:sign_problem}
Before going into the discussion about the number of the ancillary modes, we first describe a subtle problem caused by the indefinite sign of the measurement outcomes. For example, suppose one wants to implement a single-mode quadrature phase gate $V(x)=x^N$ \cite{general_npg}. One can naturally choose $f_1(x)=-x^N, K_1=1$ and get
\begin{align}
    V\star_{1,s}f_1&=x^N-(x-s)^N\\
    &=Nsx^{N-1}+\dots
\end{align}
Now one wants to choose $f_2(x)=-x^{N-1}$ and $K_2=(Ns)^{\frac{1}{N-1}}$ so that $f_2(K_2x)=-Nsx^{N-1}$, but this works only when $N-1$ is odd or $s>0$, because $K_2$ should be a real number.

Solutions of this problem would be either (a) to allow finite success probability of the gate and assume $s>0$, or (b) to prepare two ancillary states with different signs (e.g. $x^{n-1}$ and $-x^{n-1}$) and switch them depending on the sign of the measurement outcome $s$. When we choose (a), the gate is no longer deterministic, while when we take the option (b), the number of necessary ancillary modes increases. For example, in the case of the quadrature phase gate, one needs $N-2+\lfloor (N-2)/2 \rfloor$ modes instead of $N-2$ as in Ref.~\cite{cpg_gkp}.

The same problem happens also in the general cases of multi-mode non-Gaussian gates that we consider here. However, because this problem happens in all the schemes that we compare in Sec.~\ref{sec:strategies} and it only causes increase of the number of modes by a constant factor, in the rest of the paper we ignore this problem for simplicity, and assume that all measurement outcomes $s_k$ are positive.

\subsection{The a-rank of tensor}
For the coefficients of the polynomial (Eq.~(\ref{eq:v_coeffs})), we write 
\begin{align}
    V^{(j,N_j)}=V^{(j,N_j)}(s),
\end{align}
as a function of all previous measurement outcomes $s=(s_1,\dots,s_{j-1})$. We define \textit{a-rank} of $V^{(j,N_j)}(s)$ as the minimum dimension of $f_j$ that does not depend on $s$ and satisfies Eq.~(\ref{eq:ancilla_cond}). Equivalently, for a tensor $A(s)$ depending on $s$, we define \textit{a-rank} of $A(s)$ as
\begin{align}
\begin{split}
    \arank(A(s))=&\min \dim F \\
    &\mathrm{s.t.}\ \exists K(s), A(s)=F K(s)^{\otimes k}.
\end{split}
\end{align}
Using this, the number of necessary ancillary modes is upper-bounded by
\begin{align}
    \sum_{k=1}^{N-2}\arank(V^{(k-1,N+1-k)}(s)),
\end{align}
where $\arank(V^{(0,N)})=\arank(V^{(N)})=n$ (the number of the input modes).
The a-rank has the following properties.
\begin{theorem}\label{thm:arank_convex}
\begin{align}
    \arank(A(s) K(s)^{\otimes k})&\leq\arank(A(s))\\
    \arank(A_1(s)+A_2(s))&\leq\arank(A_1(s))+\arank(A_2(s))\label{eq:arank_additivity}
\end{align}
\end{theorem}
\begin{proof}
The first property directly follows from the definition.

For the second property, suppose we have decompositions
\begin{align}
    A_1(s)&=B_1K_1(s)^{\otimes k}\\
    A_2(s)&=B_2K_2(s)^{\otimes k}.
\end{align}
where $B_1$ and $B_2$ are $n_1$-dimensional and $n_2$-dimensional  tensors. We define a $(n_1+n_2)$-dimensional tensor $B_1 \oplus B_2$ as
\begin{align}
    \qty(B_1 \oplus B_2)x^{\otimes k}&=B_1x_1^{\otimes k}+B_2x_2^{\otimes k}
\end{align}
where we divide $x=(x_1,\dots,x_{n_1+n_2})^T$ into $x_1=(x_1,\dots,x_{n_1})^T$ and $x_2=(x_{n_1+1},\dots,x_{n_2})^T$, and define the matrix $K_1(s)\oplus K_2(s)$ as
\begin{align}
    K_1(s)\oplus K_2(s)=\mqty(K_1(s)&0\\0&K_2(s)).
\end{align}
Then we have
\begin{align}
    A_1(s)+A_2(s)&=\qty(B_1 \oplus B_2)\qty(K_1(s)\oplus K_2(s))^{\otimes k},
\end{align}
which leads to Eq.~(\ref{eq:arank_additivity}).
\end{proof}

\subsection{Chow decomposition and b-rank}
The following decomposition of a polynomial $A(x)$ is called \textit{Chow decomposition} \cite{chow_rank},
\begin{align}
    A(x)=\sum_{i=1}^r \prod_{j=1}^{l_i} \qty(\sum_k m^{(i)}_{jk}x_k)^{n^{(i)}_j}.\label{eq:chow_decomp}
\end{align}
Here, $r$ is called \textit{Chow rank} and denoted as $\crank(A)$. Here we allow $m^{(i)}$ in general to depend on $s$. We define \textit{b-rank} of $A$ as
\begin{align}
    \brank(A)=\sum_{i=1}^r l_i.
\end{align}
For a homogeneous polynomial $A$ of order $n$, b-rank is related to the Chow rank by
\begin{align}
    \brank(A)=n\cdot \crank(A).
\end{align}
Chow decomposition and b-rank are useful for finding a-rank of tensors, because we have the following theorem.

\begin{theorem}\label{thm:chow_decomp_arank}
The a-rank of a tensor $A(s)$ is upper-bounded as
\begin{align}
\arank(A(s))\leq \brank(A(s)).
\end{align}
\end{theorem}
\begin{proof}
Writing each term of Eq.~(\ref{eq:chow_decomp}) as $A_i(s)$ gives:
\begin{align}
    A_i(s)x^{\otimes N}=\prod_{j=1}^{l_i} \qty(\sum_k m^{(i)}_{jk}(s)x_k)^{n^{(i)}_j},
\end{align}
and we define tensors $B_i$ as
\begin{align}
    B_ix^{\otimes N}=\prod_{j=1}^{l_i} x_j^{n^{(i)}_j}.
\end{align}
Then, we have
\begin{align}
    A_i=B_i \qty(M^{(i)}(s))^{\otimes N}
\end{align}
where the matrix $M^{(i)}(s)$ is defined as
\begin{align}
    (M^{(i)}(s)x)_j=\sum_k m^{(i)}_{jk}(s)x_k.
\end{align}
Because $B_i$ does not depend on $s$, we have
\begin{align}
    \arank(A_i(s))\leq l_i.
\end{align}
Thus, from $A(s)=\sum_i A_i(s)$ and Theorem \ref{thm:arank_convex}, it follows
\begin{align}
    \arank(A(s))\leq \sum_i l_i.
\end{align}
More explicitly, defining
\begin{align}
    B(A(s))&=\bigoplus_i B_i,\label{eq:def_b}\\
    M(A(s))&=\bigoplus_i M_i,\label{eq:def_M}
\end{align}
we have
\begin{align}
    A(s)=B(A(s))M(A(s))^{\otimes N},
\end{align}
and $B(A(s))$ is not dependent on $s$ from its definition.
\end{proof}

Furthermore, because from Eq.~(\ref{eq:poly_reduction}) $V^{(j,k)}$ has the form 
\begin{align}
    V^{(j,k)}=\sum_i A_i s^{\otimes n_i},
\end{align}
the following theorem gives another upper bound of the a-rank.

\begin{theorem}\label{thm:monomial_arank}
If $Ax^{\otimes k}$ is a monomial
\begin{align}
Ax^{\otimes k}=\prod_{i=1}^{l}x_i^{n_i},
\end{align}
then for $j<k$,
\begin{align}
\arank(A s^{\otimes j})\leq l.
\end{align}
\end{theorem}
\begin{proof}
We define
\begin{align}
    q_j&=\qty(A s^{\otimes k})^{\frac{1}{k-j}}\\
    &=\qty(\prod_{i=1}^ls_{i}^{n_{i}})^{\frac{1}{k-j}}.
\end{align}
Note that here we assume $s_i>0$ (see Sec.~\ref{sec:sign_problem}). We also define
\begin{align}
    D_j(s)=\mathrm{diag}\qty(\frac{s_1}{q_j},\dots,\frac{s_l}{q_j}).\label{eq:def_d}
\end{align}
Then we get
\begin{widetext}
\begin{align}
    As^{\otimes j}(D_j(s)x)^{\otimes k-j}&=\sum_{i_1,\dots,i_k}A_{i_1,\dots,i_k}s_{i_1}\dots s_{i_j}\frac{s_{i_{j+1}}}{q_j}x_{i_{j+1}}\dots\frac{s_{i_{k}}}{q_j}x_{i_{k}}\\
    &=\frac{1}{q_j^{k-j}}\prod_{i=1}^l s_{i}^{n_{i}}\sum_{i_1,\dots,i_k}A_{i_1,\dots,i_k}x_{i_{j+1}}\dots x_{i_{k}}\\
    &=\sum_{i_1,\dots,i_k}A_{i_1,\dots,i_k}x_{i_{j+1}}\dots x_{i_{k}}\\
    &=A'x^{\otimes k-j}
\end{align}
\end{widetext}
where $A'$ is a tensor which does not depend on $s$,
\begin{align}
    A'=A 1^{\otimes j},
\end{align}
defined using a constant vector
\begin{align}
    1=\qty(1,\dots,1)^T.
\end{align}
Thus,
\begin{align}
    As^{\otimes k}=A'\qty(D_j(s)^{-1})^{\otimes N-k}
\end{align}
and we get
\begin{align}
\arank(A s^{\otimes j})\leq\arank(A')\leq l
\end{align}
from Theorem \ref{thm:arank_convex}.
\end{proof}
\begin{col}\label{col:chow_decomp}
For a tensor $A$ not depending on $s$ and $j<k$,
\begin{align}
    \arank(A s^{\otimes j})\leq \brank(A).
\end{align}
\end{col}
\begin{proof}
    Follows from Thms.~\ref{thm:arank_convex} and \ref{thm:monomial_arank}.
\end{proof}

When $f_j$ is taken to be a homogeneous polynomial of order $N-j+1$, from Eq.~(\ref{eq:poly_reduction}), Thm.~\ref{thm:arank_convex} and Col.~\ref{col:chow_decomp}, we have
\begin{align}
\arank(V^{(j+1,k)})\leq\arank(V^{(j,k)})+\brank(f_{j+1}),\end{align}
and thus
\begin{align}
\arank(V^{(j,N_j)})&\leq\sum_{i=1}^{j}\brank(f_{i}).\label{eq:strategy2_number}
\end{align}


\subsection{Decomposition into single-mode gates}\label{sec:waring}
A complementary approach to our method is to decompose the multi-mode gate into many single-mode gates \cite{cpg_gkp,exact_decomposition}. However, essentially this can be included in our measurement-based scheme without changing the non-Gaussian ancillary modes, in the following fashion.

\begin{theorem}\label{thm:waring}
If $Vx^{\otimes N}$ has a decomposition (called \textit{Waring decomposition} \cite{waring})
\begin{align}
    Vx^{\otimes N}=\sum_{i=1}^r \qty(\sum_k m^{(i)}_{k}x_k)^{N},
\end{align}
the measurement of $m(V)$ can be implemented using $r(N-2)$ non-Gaussian ancillary modes. ($r$ is called \textit{Waring rank} and denoted as $\wrank(V)$.)
\end{theorem}
\begin{proof}
Let $M$ be a matrix whose components are
\begin{align}
    M_{ik}=m^{(i)}_{k}.
\end{align}
Then when one takes
\begin{align}
    f_j(x)=\sum_{i=1}^r x_i^{N-j+1}\label{eq:waring_f}
\end{align}
and
\begin{align}
    K_j=D_j M,\label{eq:waring_K}
\end{align}
where $D_j=\mathrm{diag}(d_{j,1},\dots,d_{j,N})$ is a diagonal matrix, then $V^{(j,N_j)}$ also has a form
\begin{align}
    V^{(j,N_j)}x^{\otimes N_j}=\sum_{i=1}^r c_i(s_i) \qty(\sum_k m^{(i)}_{k}x_k)^{N-j}
\end{align}
where $c_i(s_i)$ is a coefficient only depending on $s_i$. Thus, by setting $d_{j+1,i}=(-c_i(s_i))^{\frac{1}{N-j}}$, the condition Eq.~(\ref{eq:cond_order_reduction}) is satisfied.
\end{proof}

Note that in our implementation, the number of necessary ancillary squeezed states is given by the number of the input modes $n$, where the usual decomposition into single-mode gates \cite{cpg_gkp} requires the same number of squeezed states as the number of the gates, as we also mentioned in Sec.~\ref{sec:3rd_order_examples} for specific examples. Thus, it can be more resource-efficient in cases where the Waring rank is larger than the number of the modes. Indeed, this is the case for all examples in Secs.~\ref{sec:3rd_order_examples} and \ref{sec:general_examples}.

\subsection{Strategies for minimizing the number of ancillary modes}\label{sec:strategies}
Although the minimum number of ancillary modes is still an open problem, based on the theorems proven in the preceding sections, we propose three strategies for minimizing the number of ancillary modes.
\begin{description}
    \item[Strategy I] Based on Theorem \ref{thm:chow_decomp_arank}, $f_{j}$ and $K_{j}$ are chosen by performing Chow decomposition of $V^{(j-1,N_{j-1})}$.
    \item[Strategy II] Based on Corollary \ref{col:chow_decomp}, $f_{j}$ and $K_{j}$ are chosen by performing Chow decomposition of $f_{k}$ for all $k<j$.
    \item[Strategy III] Based on Theorem \ref{thm:waring}, $f_{j}$ and $K_{j}$ are chosen by performing Waring decomposition of $V$.
\end{description}
In general, the best strategy depends on the problem to consider, as we will see in Sec.~\ref{sec:general_examples}. Below we give more details for each strategy.
\subsubsection{Strategy I}
From the construction in the proof of Thm.~\ref{thm:chow_decomp_arank}, $f_k,K_k$ are determined as
\begin{align}
    f_1&=-V, K_1=I,\\
    f_k&=B(-V^{(k-1,N-k+1)}),\\
    K_k&=M(-V^{(k-1,N-k+1)}),
\end{align}
using $B(A(s)),M(A(s))$ in Eqs.~(\ref{eq:def_b}) and (\ref{eq:def_M}). The number of the ancillary modes is given by
\begin{align}
    n+\sum_{k=1}^{N-3}\brank(V^{(k,N-k)}).\label{eq:total_ancilla_st1}
\end{align}
\subsubsection{Strategy II}
From the construction in the proofs of Thm.~\ref{thm:monomial_arank} and Col.~\ref{col:chow_decomp}, $f_k$ is chosen as
\begin{align}
    f_1&=-V,\label{eq:f_strategy2_1}\\
    f_{k+1}&=-\bigoplus_{j=1}^k B(f_j)1^{k-j+1}\label{eq:f_strategy2}\\
    &=-(f_k\oplus B(f_{k}))1,
\end{align}
using $B(A(s))$ defined in Eq.~(\ref{eq:def_b}). 
$K_k$ is also determined as
\begin{align}
    K_1&=I,\\
    K_{i+1}&=\bigoplus_{k=1}^{i}\mqty(N_j\\k)^{1/k}D_{N-k}\qty(M(f_k)s_k)M(f_k) K_k,
\end{align}
using $D_j(s_k)$ in Eq.~(\ref{eq:def_d}) and $M(A(s))$ in Eq.~(\ref{eq:def_M}). From Eq.~(\ref{eq:strategy2_number}), the total number of the ancillary modes is given by
\begin{align}
    N+\sum_{k=1}^{N-3}\sum_{i=1}^{k}\brank(f_{i}).\label{eq:total_ancilla_st2}
\end{align}

\subsubsection{Strategy III}
From the construction in the proof of Theorem \ref{thm:waring}, $f_k$ is given by Eq.~(\ref{eq:waring_f}), and $K_i$ is given by Eq.~(\ref{eq:waring_K}). The corresponding ancillary states are separable, and consist of
\begin{align}
    r_k=\sum_{i=k}^N\wrank(V^{(i)})
\end{align}
modes of $k$th-order quadrature phase states $\ket{p-kx^{k-1}=0}$ for $3\leq k\leq N$. Thus, the total number of the ancillary modes is given by
\begin{align}
    \sum_{i=3}^N (i-2)\wrank(V^{(i)}).
\end{align}
In particular, when $V$ is a homogeneous polynomial of order $N$, it is simply
\begin{align}
    (N-2)\wrank(V^{(N)}).\label{eq:total_ancilla_st3}
\end{align}
\section{Examples of higher-order gates}\label{sec:general_examples}
Here we give some examples of higher-order non-Gaussian gates and their implementations.
\subsection{A small example}\label{sec:small_example}
In order to get an intuition about how the number of non-Gaussian ancillary modes is reduced, we first consider the following specific example,
\begin{align}
    V(\bm{x})=x_1^2x_2^2+x_1^4.\label{eq:small_example}
\end{align}
Suppose one chooses 
\begin{align}
    f_1(\bm{x})&=-x_1^2x_2^2-x_1^4,\\
    K_1&=I.
\end{align}
Then one has
\begin{align}
    V^{(1,3)}(\bm{x})=2s_1x_1x_2^2+2s_2x_1^2x_2 + 4s_1x_1^3.\label{eq:small_example_v13}
\end{align}
Now, one wants to choose $f_2,K_2$ such that
\begin{align}
    V^{(1,3)}+f_2K_2^{\otimes 3}=0.\label{eq:small_example_cond}
\end{align}
One way to do this is to decompose Eq.~(\ref{eq:small_example_v13}) as
\begin{align}
    V^{(1,3)}(\bm{x})=2s_1x_1x_2^2+x_1^2(4s_1x_1+2s_2x_2)
\end{align}
(Chow decomposition of $V^{(1,3)}$), and take
\begin{align}
    f_2(\bm{x})&=-x_1x_2^2-x_3^2x_4\\
    K_2&=\mqty(2s_1&0\\0&1\\1&0\\4s_1&2s_2).
\end{align}
This corresponds to Strategy I, and we get
\begin{align}
    \arank(V^{(1,3)})\leq 4.
\end{align}
In this case, $2+4=6$ non-Gaussian ancillary modes are required in total.

Another way is to observe that Eq.~(\ref{eq:small_example}) is Chow-decomposed into $x_1^2x_2^2$ and $x_1^4$, and take
\begin{align}
    f_2(\bm{x})&=-(2x_1x_2^2+2x_1^2x_2) - 4x_3^3.
\end{align}
Then if we take
\begin{align}
    K_2=\mqty((s_1s_2^2)^{1/3}/s_1 &0&0\\0&(s_1s_2^2)^{1/3}/s_2&0\\0&0&(s_1^3)^{1/3}/s_1)\mqty(1&0\\0&1\\1&0),
\end{align}
Eq.~(\ref{eq:small_example_cond}) holds. This corresponds to Strategy II, and we get
\begin{align}
    \arank(V^{(1,3)})\leq 3.
\end{align}
In this case, $2+3=5$ non-Gaussian ancillary modes are required in total.

One can also decompose the gate into single-mode gates (Strategy III). The gate Eq.~(\ref{eq:small_example}) can be decomposed into three $x^4$ gates, because it has a Waring decomposition
\begin{align}
    V(\bm{x})=\frac{1}{2}(x_1+6^{-1/2}x_2)^4+\frac{1}{2}(x_1-6^{-1/2}x_2)^4-\frac{1}{36}x_2^4.
\end{align}
In this case, $2\times 3=6$ non-Gaussian ancillary modes are required in total.

Therefore in this case, Strategy II is optimal in terms of the number of the non-Gaussian ancillary modes. Fig.~\ref{fig:small_example} shows the corresponding scheme for implementing the gate Eq.~(\ref{eq:small_example}).
\begin{figure}[ht]
    \centering
   \includegraphics[width=1.0\linewidth]{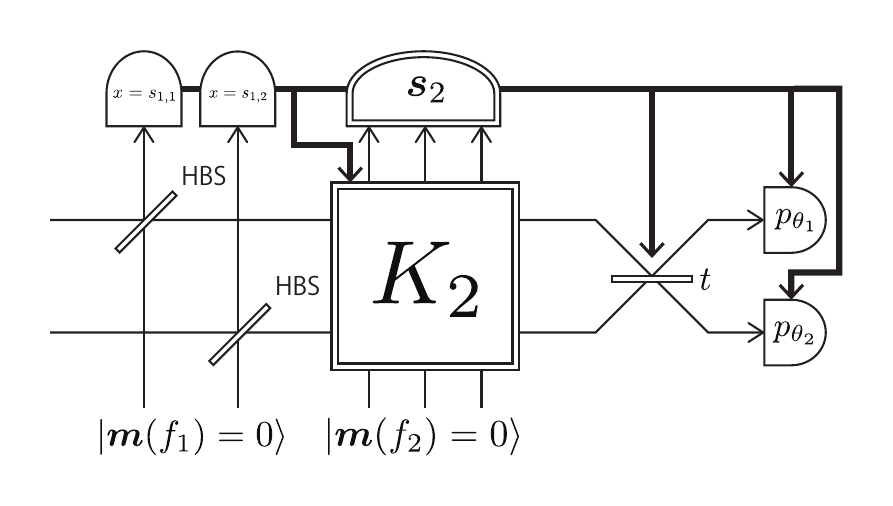}
    \caption{Implementation of the gate Eq.~(\ref{eq:small_example}), when one uses Strategy II.}
    \label{fig:small_example}
\end{figure}
\subsection{Controlled-phase gate}
We consider the following controlled phase gate,
\begin{align}
V(x_1,x_2)=x_1x_2^{N-1}.\label{eq:c_phase}
\end{align}
First, we apply Strategy I. It is straightforward to see that
\begin{align}
    B(V^{(i,N-i+1)})x^{\otimes N-i+1}=x_1x_2^{N-i}.
\end{align}
Thus, we can take 
\begin{align}
    f_i=x_1x_2^{N-i}.
\end{align}
This implementation requires only
\begin{align}
    2(N-2)\label{eq:cphase_strategy1}
\end{align}
ancillary modes. It is better than applying Strategy II.

On the other hand, if we apply Strategy III, because
\begin{align}
\wrank(V)=N
\end{align}
\cite{waring}, the decomposition into single-mode gates requires 
\begin{align}
N(N-2)\label{eq:cphase_strategy3}
\end{align}
ancillary modes. Therefore, our scheme requires a smaller number of non-Gaussian ancillary modes (Eq.~(\ref{eq:cphase_strategy3})) compared to the conventional decomposition into single-mode gates (Eq.~(\ref{eq:cphase_strategy3})).
\subsection{$C^NZ$ gate}
We consider the following $C^NZ$ gate,
\begin{align}
V_N(x_1,\dots,x_N)=x_1\dots x_{N}.\label{eq:cnz_gate}
\end{align}
We first apply Strategy I. It can be inductively shown that $V_N^{(i,N_i)}$ has the form
\begin{align}
    V_N^{(i,N_i)}x^{\otimes N_i}=\sum_{S\subset \{1,\dots,N\},|S|=N_i}a_S\prod_{j\in S}x_j\label{eq:cnz_coeffs}
\end{align}
The number of the ancillary modes is obtained by calculating the b-rank of this tensor and using Eq.~(\ref{eq:total_ancilla_st1}). Although in general it is difficult to find the minimal Chow decomposition of a polynomial, we conjecture the form of b-rank of $V^{(i,N-i+1)}$ in App.~\ref{sec:chow_decomp_cnz}. 

When we apply Strategy II, it can be inductively shown that $B(f_i)$ has the form
\begin{align}
    B(f_i) = \bigoplus_{k=1}^{\crank(f_i)}V_{N-i+1}.
\end{align}
Thus, we have
\begin{align}
    \brank(f_i)=(N-i+1)\crank(f_i),
\end{align}
and from Eqs.~(\ref{eq:f_strategy2_1}) and (\ref{eq:f_strategy2}),
\begin{align}
    \crank(f_1)&=1,\\
    \crank(f_{i+1})&=\sum_{k=1}^i c_{N-k+1,N-i}\crank(f_k),
\end{align}
where we define
\begin{align}
    c_{n,k}&=\crank(V_n 1^{\otimes n-k})\\
    &=\crank\qty(\sum_{S\subset \{1,\dots,n\},|S|=k}\prod_{j\in S}x_j).
\end{align}
In App.~\ref{sec:chow_decomp_cnz}, we conjecture an analytical form of $c_{n,k}$. Once one gets $\brank(f_i)$, the number of non-Gaussian ancillary modes can be calculated using Eq.~(\ref{eq:total_ancilla_st2}).

On the other hand, when we apply Strategy III, because 
\begin{align}
    \wrank(V_N)=2^{N-1}
\end{align}
\cite{waring}, the number of necessary ancillary modes is given by
\begin{align}
    (N-2)2^{N-1}
\end{align}
from Eq.~(\ref{eq:total_ancilla_st3}).

Table \ref{tab:cnz_ancilla} shows the comparison between Strategies I, II, and III. Calculation of each number is based on the Chow decompositions of the polynomials described in App.~\ref{sec:chow_decomp_cnz}. Strategy I gives the minimum number of the ancillary modes for $N=4,5$, whereas Strategy II is better for larger $N$. Strategy II has an advantage over Strategy III for any $N$, which corresponds to the conventional single-mode decomposition. 

\begin{table}[ht]
    \centering
    \caption{Comparison of the number of necessary non-Gaussian ancillary modes between different strategies for choosing the ancillary states, for the $C^NZ$ gate (Eq.~(\ref{eq:cnz_gate})). Calculation of each number is based on the Chow decompositions of the polynomials described in App.~\ref{sec:chow_decomp_cnz}.}
    \begin{tabular}{c|cccccc}
       $N$& 3&4&5&6&7&8\\
       \hline
        Strategy I &3&8 & 27 & 114 & 639&3936\\
       \hline
        Strategy II &3&10 & 29 & 67 & 155&333\\
        \hline
        Strategy III (conventional \cite{exact_decomposition,cpg_gkp,waring}) &4&16&48&128&320&768
    \end{tabular}
    \label{tab:cnz_ancilla}
\end{table}
\section{Conclusion \& discussion}
We have proposed a methodology including an experimentally accessible toolbox to implement, in principle, arbitrary multi-mode high-order non-Gaussian gates, relying on the concept of measurement-based quantum gates. For the 3rd-order cases, we have introduced a generalized linear coupling as a generalization of the technique used in the implementation of a CPG, which includes degrees of freedom that allow virtually applying adaptive Gaussian operations to the ancillary state. For the higher-order cases, we have proposed an implementation based on cascaded generalized linear couplings and feedforwards. We have also proposed a heuristic algorithm to reduce the number of non-Gaussian ancillary modes, based on Chow decomposition of polynomials. Our scheme does not require adaptive preparation of non-Gaussian states depending on previous measurement outcomes, and it requires only offline preparation of fixed non-Gaussian states, together with adaptive linear optics, unlike previous proposals \cite{multimode_coupling}.

We applied our method to some important examples, namely the cubic-QND gate, the CV Toffoli gate, the controlled-phase gate and the $C^nZ$ gate. In all cases, we observe that our method requires a smaller number of ancillary modes compared to conventional methods that decompose the gates into multiple single-mode gates \cite{cpg_gkp}. For higher-order cases, we found that different strategies for reducing the number of ancillary modes lead to different performances, and the best strategy to adopt depends on the types of gates. Thus, though general and systematic, our approach provides sufficient degrees of freedom for further optimization by refining the algorithms.

Our results enable a more resource-efficient and experimentally feasible implementation of CV gates compared to conventional schemes. This will accelerate the progress toward fault-tolerant universal quantum information processing, especially with light, and it highlights the inherent computational potential that CV quantum systems have. As a future extension of our work, methods for generating the multi-mode non-Gaussian ancillary states needed for our scheme could be explored, potentially through optimization of Fock-basis coefficients, which has been discussed for the case of the cubic QND gate \cite{multimode_coupling} and experimentally demonstrated for the CPG \cite{sakaguchi_cpg}.
\begin{acknowledgments}
This work was partly supported by JST [Moonshot R\&D][Grant No.~JPMJMS2064], JSPS KAKENHI (Grant No.~23KJ0498), UTokyo Foundation, and donations from Nichia Corporation. PvL acknowledges funding from the BMBF in Germany (QR.X, PhotonQ, QuKuK, QuaPhySI), from the EU’s HORIZON Research and Innovation Actions (CLUSTEC), and from the Deutsche Forschungsgemeinschaft
(DFG, German Research Foundation) - Project-ID 429529648 -
TRR 306 QuCoLiMa (“Quantum Cooperativity of Light and Matter”).
\end{acknowledgments}
\bibliography{main.bib}
\appendix
\begin{widetext}
    
\section{Decomposition of arbitrary gates into quadrature gates}\label{sec:decomp_quad_appendix}
In this section, we explicitly give a decomposition of an arbitrary Hamiltonian
\begin{align}
H(x_1,\dots,x_n)=\sum_i \qty(c_i\prod_k x_k^{m_{ik}}p_k^{n_{ik}}+c_i^*\prod_k p_k^{n_{ik}}x_k^{m_{ik}})\label{eq:arbitrary_hamiltonian}
\end{align}
into quadrature gates
\begin{align}
    \prod_k x_k^{m_{k}},\label{eq:quad_gate}
\end{align}
using Trotter-Suzuki approximation \cite{trotter_suzuki,decomp_gate}.
The goal here is to express the Hamiltonian Eq.~(\ref{eq:arbitrary_hamiltonian}) using sum (splitting) and commutators of quadrature gates Eq.~(\ref{eq:quad_gate}). Because Eq.~(\ref{eq:arbitrary_hamiltonian}) can be rewritten as
\begin{align}
    H(x_1,\dots,x_n)=\sum_i \qty(\mathrm{Re}(c_i)\Bigl\{\prod_k x_k^{m_{ik}},\prod_{k'} p_{k'}^{n_{i{k'}}}\Bigr\}+i\mathrm{Im}(c_i)\Bigl[\prod_k x_k^{m_{ik}},\prod_{k'} p_{k'}^{n_{i{k'}}}\Bigr]),
\end{align}
where $\{\cdot\}$ is an anti-commutator and $[\cdot]$ is a commutator, it suffices to give a decomposition of the term
\begin{align}
    \Bigl\{\prod_k x_k^{m_{ik}},\prod_{k'} p_{k'}^{n_{i{k'}}}\Bigr\}.
\end{align}

For doing this, we generalize the following single-mode result in Ref.~\cite{decomp_gate},
\begin{align}
\{x^M,p^N\}=-\frac{2i}{(N+1)(M+1)}[x^{M+1},p^{N+1}]-\frac{1}{N+1}\sum_{k=1}^{N-1}[p^{N-k},[x^M,p^k]],
\end{align}
to the multi-mode case. Note that here we rewrite the original equation where they use the convention $[x,p]=i/2$, with our convention $[x,p]=i$. We write
\begin{align}
\bm{x}=(x_1,\dots,x_n), \bm{p}=(p_1,\dots,p_n),
\end{align}
and introduce a vector exponent notation: for $\bm{M}=(M_1,\dots,M_n)$,
\begin{align}
\bm{x}^{\bm{M}}=x_1^{M_1}\dots x_n^{M_n}.
\end{align}
Using the equation
\begin{align}
[x^M,p^N]=iM\sum_{k=0}^{N-1}p^{k}x^{M-1}p^{N-k-1}
\end{align}
 \cite{decomp_gate}, we have
\begin{align}
    [\bm{x}^{\bm{M}},\bm{p}^{\bm{N}}]
    &=[x_1^{M_1},p_1^{N_1}]x_2^{M_2}\dots x_n^{M_n}p_2^{N_2}\dots p_n^{N_n}+\dots+p_1^{N_1}\dots p_{n-1}^{N_{n-1}}x_1^{M_1}\dots x_{n-1}^{M_{n-1}}[x_n^{M_n},p_n^{N_n}]\\
    &=\sum_{j=1}^n\bm{p}^{\bm{N}_{<j}}[x_j^{M_j},p_j^{N_j}]\bm{x}^{\bm{M}_{/j}}\bm{p}^{\bm{N}_{>j}}\\
    &=\sum_{j=1}^n iM_j\sum_{k=0}^{N_j-1}\bm{p}^{\bm{N}_{<j}}p_j^{k}\bm{x}^{\bm{M}_{/j}}x_j^{M_j-1}p_j^{N_j-k-1}\bm{p}^{\bm{N}_{>j}}\\
    &=\sum_{j=1}^n iM_j\sum_{k=0}^{N_j-1}\bm{p}^{\bm{N}_{<j}+k\bm{e}_j}\bm{x}^{\bm{M}-\bm{e}_j}\bm{p}^{\bm{N}_{>j}+(N_j-k-1)\bm{e}_j}\\
    &=\frac{i}{2}\sum_{j=1}^n M_j\sum_{k=0}^{N_j-1}\qty(\bm{p}^{\bm{N}_{<j}+k\bm{e}_j}\bm{x}^{\bm{M}-\bm{e}_j}\bm{p}^{\bm{N}_{>j}+(N_j-k-1)\bm{e}_j}+\mathrm{H.c.})\\
    &=\sum_{j=1}^n \qty(\frac{iM_jN_j}{2}\qty{\bm{x}^{\bm{M}-\bm{e}_j},\bm{p}^{\bm{N}-\bm{e}_j}}+\frac{iM_j}{2}\sum_{k=0}^{N_j-1}\qty[\bm{p}^{\bm{N}_{>j}+(N_j-k-1)\bm{e}_j},\qty[\bm{x}^{\bm{M}-\bm{e}_j},\bm{p}^{\bm{N}_{<j}+k\bm{e}_j}]])
\end{align}
Here $\bm{e}_j$ is a vector whose $j$-th component is 1 and the others are 0. $\bm{N}_{<j}$ is a vector whose $1,\dots,j-1$-th components are the same as $\bm{N}$ and the others are 0. $\bm{N}_{/j}$ is a vector whose $j$-th component is 0 and the others are the same as $\bm{N}$.

Thus, the decomposition of $\qty{\bm{x}^{\bm{M}},\bm{p}^{\bm{N}}}$ can be calculated as
\begin{align}
\begin{split}
\qty{\bm{x}^{\bm{M}},\bm{p}^{\bm{N}}}=&-\frac{2i}{(M_1+1)(N_1+1)}\Bigl(\qty[\bm{x}^{\bm{M}+\bm{e}_1},\bm{p}^{\bm{N}+\bm{e}_1}]-\frac{i(M_1+1)}{2}\sum_{k=0}^{N_{1}}\qty[\bm{p}^{\bm{N}-k\bm{e}_{1}},\qty[\bm{x}^{\bm{M}},\bm{p}^{k\bm{e}_{1}}]]\\
&-\sum_{j=2}^n\frac{iM_jN_j}{2}\qty{\bm{x}^{\bm{M}+\bm{e}_1-\bm{e}_j},\bm{p}^{\bm{N}+\bm{e}_1-\bm{e}_j}}-\sum_{j=2}^n\frac{iM_{j}}{2}\sum_{k=0}^{N_{j}-1}\qty[\bm{p}^{\bm{N}_{>j}+(N_{j}-k-1)\bm{e}_{j}},\qty[\bm{x}^{\bm{M}+\bm{e}_1-\bm{e}_{j}},\bm{p}^{\bm{N}_{<j}+\bm{e}_1+k\bm{e}_{j}}]]\Bigr)
\end{split}\label{eq:quad_decomp}
\end{align}
The anticommutator in the third term can be recursively decomposed using the same equation. This recursion stops after a finite number of steps, because the exponents $\bm{M}+\bm{e}_1-\bm{e}_j, \bm{N}+\bm{e}_1-\bm{e}_j$ have the same sum of the components as $\bm{M},\bm{N}$, while $M_1,N_1$ keep increasing by one for each step. Note that, in the final expression after the recursive application of Eq.~(\ref{eq:quad_decomp}), some of the terms in the sum could be combined because the same exponent may appear multiple times.

Here we give some examples.
\begin{align}
\begin{split}
&\{x_1x_2,p_1p_2\}\\&=- \frac{i}{2}[x_1^2x_2,p_1^2p_2]- \frac{2 i}{9}[x_1^3,p_1^3]+\frac{1}{2}[p_2,[x_1x_2,p_1]]+\frac{1}{3}[p_1,[x_1^2,p_1]]\\
&=- \frac{i}{2}[x_1^2x_2,p_1^2p_2]-\frac{2 i}{9}[x_1^3,p_1^3]+\frac{7}{6},
\end{split}\label{eq:quad_decomp_example1}\\
\begin{split}
&\{x_1^3x_2,p_1^2p_2^2\}\\
&=- \frac{i}{6}[x_1^4x_2,p_1^3p_2^2]- \frac{i}{10}[x_1^5,p_1^4p_2]+\frac{1}{3}[p_1p_2^2,[x_1^3x_2,p_1]]+\frac{1}{3}[p_2^2,[x_1^3x_2,p_1^2]]\\&+\frac{1}{3}[p_2,[x_1^4,p_1^3]]+\frac{1}{4}[p_1^2p_2,[x_1^4,p_1]]+\frac{1}{4}[p_1p_2,[x_1^4,p_1^2]],
\end{split}\\
\begin{split}
&\{x_1x_2x_3,p_1p_2p_3\}\\
&=- \frac{i}{2}[x_1^2x_2x_3,p_1^2p_2p_3]- \frac{i}{18}[x_1^3x_3,p_1^3p_3]- \frac{i}{16}[x_1^4,p_1^4]- \frac{i}{18}[x_1^3x_2,p_1^3p_2]\\&+\frac{1}{2}[p_2p_3,[x_1x_2x_3,p_1]]+\frac{1}{3}[p_3,[x_1^2x_3,p_1^2]]+\frac{1}{12}[p_1p_3,[x_1^2x_3,p_1]]+\frac{1}{4}[p_1^2,[x_1^3,p_1]]\\&+\frac{1}{12}[p_1p_2,[x_1^2x_2,p_1]]+\frac{1}{12}[p_2,[x_1^2x_2,p_1^2]].
\end{split}
\end{align}

Note that this decomposition usually includes less nesting of commutators compared to known decompositions into single-mode and Gaussian entangling gates \cite{seth_decomp,decomp_gate}. For example, a naive application of the method in Ref.~\cite{decomp_gate} to the Hamiltonian in Eq.~(\ref{eq:quad_decomp_example1}) leads to
\begin{align}
\{x_1x_2,p_1p_2\}
&=\frac{1}{72}[p_2^2,[p_1^2,[x_2^3,[x_1^3,p_1p_2]]+\frac{1}{2},
\end{align}
which includes a deeply nested commutator, requiring a higher number of gates for getting a certain accuracy of the gate.
Thus, our direct decomposition into multi-mode quadrature gates, combined with our measurement-based direct implementation of multi-mode quadrature gates, gives a more efficient way to implement high-order multi-mode gates.

\end{widetext}
\section{Simultaneous measurement of linear quadrature operators using linear optics}\label{sec:linear_measurement}
In this section, we show the following theorem.
\begin{theorem}\label{thm:linear_meas}
A set of $n$ linear combinations of quadrature operators
\begin{align}
\bm{q}=\bm{p}+A\bm{x},\label{eq:quad_specialcase}
\end{align}
where $A$ is a real symmetric matrix: $A_{ij}=A_{ji}$, can be simultenously measured only using beamsplitters and homodyne measurements.
\end{theorem}
    
\begin{proof}
Because the matrix $A=(A_{ij})$ is symmetric, it can be diagonalized as
\begin{align}
    A= O^T D O
\end{align}
using an orthogonal matrix $O$ and a diagonal matrix
\begin{align}
    D=\mathrm{diag}(\lambda_1,\dots,\lambda_n).
\end{align}
Thus, Eq.~(\ref{eq:quad_specialcase}) can be written as
\begin{align}
    \hat{\boldsymbol{q}}&=\hat{\boldsymbol{p}}+O^T D O \hat{\boldsymbol{x}}\\
    &=O^T(O\hat{\boldsymbol{p}}+ D O \hat{\boldsymbol{x}})\\
    &=O^T(\hat{\boldsymbol{p}}'+ D \hat{\boldsymbol{x}}')\label{eq:q_measurement}
\end{align}
where we denote $\hat{\boldsymbol{p}}'=O\hat{\boldsymbol{p}}$ and $\hat{\boldsymbol{x}}'=O\hat{\boldsymbol{x}}$.

Based on Eq.~(\ref{eq:q_measurement}), the protocol to measure $q_i$ is the following. First, a multi-mode beamsplitter corresponding to $O$ is applied, then homodyne measurements of operators
\begin{align}
\hat{p}_{\theta_i}=\hat{p}_i\cos\theta_i+\hat{x}_i\sin\theta_i
\end{align}
are performed on each mode. The phases $\theta_i$ are determined as
\begin{align}
    \theta_i=\mathrm{arctan}(\lambda_i).
\end{align}
Finally, after obtaining the homodyne measurement outcomes
\begin{align}
    \hat{p}_{\theta_i}=y_i,
\end{align}
values of $\hat{q}_i$ can be obtained by classical post-processing:
\begin{align}
    q_i=\sum_j O^T_{ij} y_j \cos\theta_j.
\end{align}
\end{proof}

As a generalization of the Theorem \ref{thm:linear_meas}, we have the following theorem.
\begin{theorem}\label{thm:gen_linear_meas}
A set of $n$ commuting linear combinations of quadrature operators
\begin{align}
\bm{q}=B\bm{p}+C\bm{x},\label{eq:q_BC}
\end{align}
can be simultenously measured only using beamsplitters and homodyne measurements.
\end{theorem}
\begin{proof}
Equation (\ref{eq:q_BC}) can be rewritten as
\begin{align}
    \hat{q}_i=\sum_{j} (A_{ij}\hat{a}_j+ A^*_{ij}\hat{a}_j^\dagger),\label{eq:q_A}
\end{align}
where $A_{ij}=\frac{1}{\sqrt{2}}(C_{ij}-iB_{ij})$. From the condition that all $\hat{q}_i$ commute, the matrix $A=(A_{ij})$ should satisfy
\begin{align}
    (AA^\dagger)^T=AA^\dagger.
\end{align}
As $AA^\dagger$ is a symmetric real matrix, it can be diagonalized as
\begin{align}
    AA^\dagger = O' D_1 O'^T
\end{align}
where $O'$ is an orthogonal matrix, and $D_1$ is a real diagonal matrix. Thus, $A$ has a singular value decompositon of the form
\begin{align}
    A=O' D_1 V \label{eq:A_decomp1}
\end{align}
where $V$ is a unitary matrix.

Now we use the fact that any unitary matrix $V$ can be decomposed as
\begin{align}
    V=O'' \Phi O\label{eq:A_decomp2}
\end{align}
where $O'',O$ are orthogonal matrices, and $\Phi$ is a diagonal matrix having complex diagonal elements $e^{i\phi_i}$ \cite{unitary_decomp}. Combining Eqs.~(\ref{eq:A_decomp1}) and (\ref{eq:A_decomp2}), we obtain
\begin{align}
    A=O' D_1 O'' \Phi O.\label{eq:A_decomp3}
\end{align}
Thus, by denoting the real matrix $O' D_1 O''$ as $P$, Eq.~(\ref{eq:q_A}) can be written as
\begin{align}
    \hat{\boldsymbol{q}}=P(\Phi O \hat{\boldsymbol{a}}+\Phi^* O^* \hat{\boldsymbol{a}}^\dagger)\label{eq:A_decomp_final}
\end{align}
where $\hat{\boldsymbol{q}}=(\hat{q}_i),\hat{\boldsymbol{a}}=(\hat{a}_i)$ are vectors of operators.

Equation (\ref{eq:A_decomp_final}) means the simultaneous measurement of $\hat{q}_i$ can be achieved by using passive beamsplitters corresponding $O$ followed by homodyne measurements in the phase determined by $\phi_i$. Then one can use classical post-processing of measurement outcomes corresponding to the matrix $P$ to get values of $\hat{q}_i$.
\end{proof}
    
\section{Proof of Theorem \ref{thm:diamond_relation}}\label{sec:generalized_coupling_appx}

In this section, we give a proof of Theorem \ref{thm:diamond_relation}. We consider the scheme in Fig.~\ref{fig:building_block}. We define $n\times n$ matrix $T$, $n'\times n'$ matrix $T'$, $n'\times n$ matrix $R$ as
\begin{align}
T&=\begin{cases}
    \mathrm{diag}(t_1,\dots,t_n) & \mathrm {if\ } n\leq n'\\
    \mathrm{diag}(t_1,\dots,t_{n'},1,\dots,1) & \mathrm {otherwise}
\end{cases}\\
T'&=\begin{cases}
    \mathrm{diag}(t_1,\dots,t_{n'}) & \mathrm {if\ } n'\leq n\\
    \mathrm{diag}(t_1,\dots,t_{n},1,\dots,1) & \mathrm {otherwise}
\end{cases}\\
R&=\begin{cases}
    \mathrm{diag}(r_1,\dots,r_n) & \mathrm {if\ } n\leq n'\\
    \mathrm{diag}(r_1,\dots,r_{n'}) & \mathrm {otherwise}
\end{cases}
\end{align}
Then, the relation between input and output quadratures can be written as
\begin{align}
    x_{-}&=ROx-T'O'x',\\
    p_{-}&=ROp-T'O'p',\\
    x_{+}&=O^T(TOx+R^TO'x'),\\
    p_{+}&=O^T(TOp+R^TO'p').
\end{align}
Then we have
\begin{widetext}
\begin{align}
\begin{split}
    &O^TTOm(x,p;f)+O^TR^TO'm(x',p';g)=p_{+}-\frac{\partial}{\partial x_+} \qty[f(O^T(TOx_+ + R^Tx_-))+g(O^{\prime T}(ROx_+ - T'x_-))].
\end{split}\label{eq:m_relation_orig}
\end{align}
If one measures $x_-$ and gets $O'^{T} T' x_-=s$,
\begin{align}
\begin{split}
    &f(O^T(T Ox_+ + Rx_-))+g(O^{\prime T}(ROx_+ - Tx_-))=f(P(K)x_+ + K^T s)+g\qty(K P(K) x_+ - s),
\end{split}\label{eq:K_PK_rel}
\end{align}
\end{widetext}
where we used
\begin{align}
    K&=O^{\prime T}RT^{-1}O,\\
    P(K)&=(I+K^TK)^{-1/2}\\
    &=O^{T}TO.
\end{align}
Using Eq.~(\ref{eq:K_PK_rel}), Eq.~(\ref{eq:m_relation_orig}) can be rewritten as
\begin{align}
P(K)m(x,p;f)+KP(K)m(x',p';g)=m(x_+,p_+;f\diamond_{K,s}g),
\end{align}
which is the statement of Theorem \ref{thm:diamond_relation}.

\section{Proof of Theorem \ref{thm:diamond_star_seq}}\label{sec:diamond_star_seq}
We introduce the notation
\begin{align}
    \qty[f\circ(A,b)](x)=f(Ax+b).
\end{align}
Then we have
\begin{widetext}
\begin{align}
((f \circ(A,b)) \diamond_{KA,\bm{s}} g)(\bm{x})&=f(AP(KA)\bm{x} + AA^TK^T \bm{s}+b)+g(KA P(KA)\bm{x} - \bm{s})\\
&=(f\star_{K,KAb+(I+KAA^TK^T)\bm{s}}g)\circ(A,b)\circ(P(KA),A^TK^T \bm{s}).
\end{align}
\end{widetext}
Theorem \ref{thm:diamond_star_seq} can be proven by repeating this transformation.

\section{Chow decomposition for $C^NZ$ gate}\label{sec:chow_decomp_cnz}
Here we consider the Chow decomposition of the following polynomial
\begin{align}
    V_{n,k}^{a}x^{\otimes k}=\sum_{S\subset \{1,\dots,n\},|S|=k}a_S\prod_{j\in S}x_j.\label{eq:arbitrary_polynomial}
\end{align}
We first consider the case where $a_S=1$ for all $S$:
\begin{align}
    V_{n,k}x^{\otimes k}=\sum_{S\subset \{1,\dots,n\},|S|=k}\prod_{j\in S}x_j.
\end{align}
For $n,k\in \mathbb{N}$, let
\begin{align}
    P(n,k)=\qty{(p_j)\in \mathbb{N}^k |p_j\geq 1, \sum_{j=1}^k p_j=n}
\end{align}
be a set of all partitions of $n$ into $k$ pieces. For $(p_j) \in P(n,k)$, we define a polynomial $r\qty((p_j))$ as
\begin{align}
r\qty((p_j))(\bm{x})=\prod_{j=1}^{k}\sum_{m\in R_j}x_{m},
\end{align}
where
\begin{align}
    R_1&=[1,p_1],\\
    R_j&=(\sum_{l=1}^{j-1} p_l,\sum_{l=1}^{j} p_l] &(1<j\leq k).
\end{align}
We define $Q(n,l)\subset P(n,2l+1)$ as
\begin{align}
    Q(n,l)=\qty{(p_j)\in P(n,2l+1) | p_{2m}=1, m=1,\dots l}.
\end{align}
For example,
\begin{align}
    r\qty((2,1,3))=(x_1+x_2)x_3(x_4+x_5+x_6).
\end{align}
and
\begin{align}
\begin{split}
    Q(7,2)=\{&(3,1,1,1,1),(1,1,3,1,1),(1,1,1,1,3),\\
    &(2,1,2,1,1),(2,1,1,1,2),(1,1,2,1,2)\}.
\end{split}
\end{align}
Then we have the following theorem.
\begin{theorem}\label{thm:chow_decomp_poly}
The following expressions
\begin{align}
    &\sum_{(p_j) \in Q(n,l)}r\qty((p_j)) & (k=2l+1)\label{eq:chow_decomp_odd}\\
    &\sum_{m=k}^{n}\sum_{(p_j) \in Q(m-1,l-1)}r\qty((p_j))x_{m} & (k=2l)\label{eq:chow_decomp_even}
\end{align}
give a Chow decomposition of $V_{n,k}$.
\end{theorem}
\begin{proof}
We first prove the case of $k=2l+1$. It suffices to show that for any $S\subset \{1,\dots,n\}$ such that $|S|=k$, there exists a unique $(p_j) \in Q(n,l)$ such that expansion of $r\qty((p_j))$ includes $\prod_{j\in S}x_j$. In fact, such $(p_j)$ is given by
\begin{align}
\begin{split}
    p_{1}&=s_2-1\\
    p_{2l+1}&=s_{2l+2}-s_{2l}-1\\
    p_{2l}&=1.
\end{split}
\end{align}
where $(s_j)$ is a sequence of all elements of $S$ in an ascending order $s_1\leq\dots\leq s_k$.

For the case of $k=2l$, we have
\begin{align}
    \prod_{j\in S}x_J=\qty(\prod_{j\in S\setminus \qty{s_k}}x_j) x_{s_k},
\end{align}
where $s_k$ is the maximal element of $S$. By performing the Chow decomposition to the first part, we get Eq.~(\ref{eq:chow_decomp_even}).
\end{proof}

Below we show some examples of the obtained Chow decompositions,
\begin{widetext}
\begin{align}
\begin{split}
    V_{6,3}=&x_1x_2(x_3+x_4+x_5+x_6)+(x_1+x_2)x_3(x_4+x_5+x_6)\\&+(x_1+x_2+x_3)x_4(x_5+x_6)+(x_1+x_2+x_3+x_4)x_5x_6,
\end{split}\label{eq:v63}\\
\begin{split}
    V_{6,4}=&x_1x_2x_3x_4+(x_1+x_2)x_3x_4x_5+x_1x_2(x_3+x_4)x_5+(x_1+x_2+x_3)x_4x_5x_6\\&+(x_1+x_2)x_3(x_4+x_5)x_6+x_1x_2(x_3+x_4+x_5)x_6,
\end{split}\\
\begin{split}
    V_{7,5}=&(x_1+x_2+x_3)x_4x_5x_6x_7+x_1x_2(x_3+x_4+x_5)x_6x_7+x_1x_2x_3x_4(x_5+x_6+x_7)\\&+(x_1+x_2)x_3(x_4+x_5)x_6x_7+(x_1+x_2)x_3x_4x_5(x_6+x_7)+x_1x_2(x_3+x_4)x_5(x_6+x_7).
\end{split}
\end{align}
\end{widetext}

By using binomial coefficients we can write
\begin{align}
    \qty|Q(n,l)|=\mqty(n-l-1\\l).
\end{align}
Thus, Eq.~(\ref{eq:chow_decomp_odd}) has $\mqty(n-l-1\\l)$ terms. For Eq.~(\ref{eq:chow_decomp_even}), the number of terms is
\begin{align}
    \sum_{m=k}^n\mqty(m-l-1\\l-1)&=\mqty(n-l\\l).
\end{align}
Thus, we have the following corollary.
\begin{col}
\begin{align}
&\crank(V_{n,2l}),\crank(V_{n+1,2l+1})\leq\mqty(n-l\\l),\\
&\brank(V_{n,2l})\leq2l\mqty(n-l\\l),\\
&\brank(V_{n+1,2l+1})\leq(2l+1)\mqty(n-l\\l).
\end{align}
\end{col}
We conjecture that this gives the minimal Chow decomposition.
\begin{cj}
\begin{align}
&\crank(V_{n,2l}),\crank(V_{n+1,2l+1})=\mqty(n-l\\l),\\
&\brank(V_{n,2l})=2l\mqty(n-l\\l),\\
&\brank(V_{n+1,2l+1})=(2l+1)\mqty(n-l\\l).
\end{align}
\end{cj}

In the case where $a_S\neq 1$ in Eq.~(\ref{eq:arbitrary_polynomial}), we can modify the Chow decomposition in Theroem \ref{thm:chow_decomp_poly} by expanding the terms so that each term only has one non-monomial factor. We define expression $r_{i}\qty((p_j))$ as a partial expansion of $r\qty((p_j))$ except for the factor corresponding to $p_i$, with additional coefficients $a_S$ to each term,
\begin{align}
   r_{i}\qty((p_j))=\sum_{(m_j)\in \prod_{j\neq i} R_j}\qty(\sum_{m \in R_i}a_{\{m_j|j\neq i\}\cup \{m\}}x_{m})\prod_{j\neq i}x_{m_j}.
\end{align}
Then $V^{a}_{n,k}$ has a Chow decomposition in the following form.
\begin{col}
The following expressions
\begin{align}
\begin{cases}
    \sum_{(p_j) \in Q(n,l)}r_{\mathrm{argmax}_i p_i}\qty((p_j)) & (k=2l+1)\\
    \sum_{m=k}^{n}\sum_{(p_j) \in Q(m-1,l-1)}r_{\mathrm{argmax}_i p_i}\qty((p_j))x_{m} & (k=2l)
\end{cases}
\end{align}
give a Chow decomposition of $V^{a}_{n,k}$.
\end{col}

For example,
\begin{widetext}
\begin{align}
\begin{split}
r_3((2,1,3))=&x_1x_3(a_{134}x_4+a_{135}x_5+a_{136}x_6)+x_2x_3(a_{234}x_4+a_{235}x_5+a_{236}x_6)
\end{split}
\end{align}
and
\begin{align}
\begin{split}
V^{a}_{6,3}=&x_1x_2(a_{123}x_3+a_{124}x_4+a_{125}x_5+a_{126}x_6)+x_1x_3(a_{134}x_4+a_{135}x_5+a_{136}x_6)\\&+x_2x_3(a_{234}x_4+a_{235}x_5+a_{236}x_6)+(a_{145}x_1+a_{245}x_2+a_{345}x_3)x_4x_5\\&+(a_{146}x_1+a_{246}x_2+a_{346}x_3)x_4x_6+(a_{156}x_1+a_{256}x_2+a_{356}x_3+a_{456}x_4)x_5x_6.
\end{split}
\end{align}
\end{widetext}
This may then also be compared with Eq.~(\ref{eq:v63}).


\end{document}